\documentclass[a4paper,USenglish,cleveref]{lipics-v2019} 
\nolinenumbers
\hideLIPIcs
\makeatletter
\renewcommand\@oddfoot{
	\hfil
	\rlap{%
		\vtop{%
			\vskip10mm
			\colorbox{lipicsYellow}
			{\@tempdima\evensidemargin
				\advance\@tempdima1in
				\advance\@tempdima\hoffset
				\hb@xt@\@tempdima{%
					\textcolor{lipicsGray}{\normalsize\sffamily
						\bfseries\quad
						\expandafter\textsolittle\expandafter{%
							arXiv.org}}%
					\strut\hss}}}}
}
\makeatother

\sffamily\DeclareFontShape{T1}{lmss}{bx}{sc} { <-> ssub * cmr/bx/sc }{}

\newcommand{\card}[1]{\|#1\|}
\newcommand{\assignment}{\sigma}
\newcommand{\assignmenttwo}{\tau}
\newcommand{\graphnumber}{\xi}

\newcounter{tempcounter}
\usepackage{refcount}

\usepackage{booktabs}
\usepackage{xspace}
\usepackage{tikz}

\usepackage{mathtools,bm}

\renewcommand{\setminus}{-}
\newcommand{\alternativetextsc}[1]{\textup{\textsf{\textsc{#1}}}}
\newcommand{\N}{\mathbb{N}}
\newcommand{\Z}{\mathbb{Z}}
\newcommand{\CNF}{\ensuremath{\alternativetextsc{CNF}}\xspace}
\newcommand{\ThreeCNF}{\ensuremath{\alternativetextsc{3CNF}}\xspace}
\newcommand{\EThreeCNF}{\ensuremath{\alternativetextsc{3CNF}}\xspace}
\newcommand{\Sat}{\ensuremath{\alternativetextsc{Sat}}\xspace}
\newcommand{\UnSat}{\ensuremath{\alternativetextsc{UnSat}}\xspace}
\newcommand{\ThreeSat}{\ensuremath{\alternativetextsc{3Sat}}\xspace}
\newcommand{\EThreeSat}{\ensuremath{\alternativetextsc{3Sat}}\xspace}
\newcommand{\ThreeUnSat}{\ensuremath{\alternativetextsc{3UnSat}}\xspace} 
\newcommand{\MinimalUnSat}{\ensuremath{\alternativetextsc{Minimal}\-\alternativetextsc{Un}\-\alternativetextsc{Sat}}\xspace} 
\newcommand{\MinimalThreeUnSat}{\ensuremath{\alternativetextsc{Minimal}\-\alternativetextsc{3UnSat}}\xspace} 
\newcommand{\ETwoSat}{\ensuremath{\alternativetextsc{2Sat}}\xspace}
\newcommand{\ThreeCol}{\ensuremath{\alternativetextsc{3Color}\-\alternativetextsc{ability}}\xspace}
\newcommand{\VertexMinimalThreeUnCol}{\ensuremath{\alternativetextsc{Vertex}\-\alternativetextsc{Minimal}\-\alternativetextsc{3UnColor}\-\alternativetextsc{ability}}\xspace}
\newcommand{\Criticality}{\ensuremath{\alternativetextsc{Criticality}}\xspace}
\newcommand{\VertexCriticality}{\ensuremath{\alternativetextsc{Vertex}\-\alternativetextsc{Criticality}}\xspace}
\newcommand{\Stability}{\ensuremath{\alternativetextsc{Stability}}\xspace}
\newcommand{\VertexStability}{\ensuremath{\alternativetextsc{Vertex}\-\alternativetextsc{Stability}}\xspace}
\newcommand{\Unfrozenness}{\ensuremath{\alternativetextsc{Unfrozenness}}\xspace}
\newcommand{\VertexUnfrozenness}{\ensuremath{\alternativetextsc{Vertex}\-\alternativetextsc{Unfrozenness}}\xspace}
\newcommand{\TwoWayStability}{\ensuremath{\alternativetextsc{TwoWayStability}}\xspace}
\newcommand{\VertexTwoWayStability}{\ensuremath{\alternativetextsc{Vertex}\-\alternativetextsc{TwoWayStability}}\xspace}
\newcommand{\betaStability}{\ensuremath{\beta\text{-}\alternativetextsc{Stability}}\xspace}
\newcommand{\betaVertexStability}{\ensuremath{\beta\text{-}\alternativetextsc{Vertex}\-\alternativetextsc{Stability}}\xspace}
\newcommand{\omegaVertexStability}{\ensuremath{\omega\text{-}\alternativetextsc{Vertex}\-\alternativetextsc{Stability}}\xspace}
\newcommand{\alphaVertexStability}{\ensuremath{\alpha\text{-}\alternativetextsc{Vertex}\-\alternativetextsc{Stability}}\xspace}
\newcommand{\StableUnSat}{\ensuremath{\alternativetextsc{Stable}\-\alternativetextsc{UnSat}}\xspace}
\newcommand{\StableThreeUnSat}{\ensuremath{\alternativetextsc{Stable}\-\alternativetextsc{3UnSat}}\xspace}
\newcommand{\FourCNF}{\ensuremath{\alternativetextsc{4CNF}}\xspace}
\newcommand{\SixCNF}{\ensuremath{\alternativetextsc{6CNF}}\xspace}
\newcommand{\StableCNF}{\ensuremath{\alternativetextsc{Stable}\-\alternativetextsc{CNF}}\xspace}
\newcommand{\StableThreeCNF}{\ensuremath{\alternativetextsc{Stable}\-\alternativetextsc{3CNF}}\xspace}
\newcommand{\StableFourCNF}{\ensuremath{\alternativetextsc{Stable}\-\alternativetextsc{4CNF}}\xspace}
\newcommand{\StableSixCNF}{\ensuremath{\alternativetextsc{Stable}\-\alternativetextsc{6CNF}}\xspace}
\newcommand{\coG}{\ensuremath{\overline{G}}}
\renewcommand{\P}{\ensuremath{\textnormal{P}}\xspace}
\newcommand{\NP}{\ensuremath{\textnormal{NP}}\xspace}
\newcommand{\coNP}{\ensuremath{\textnormal{coNP}}\xspace}
\newcommand{\DP}{\ensuremath{\textnormal{DP}}\xspace}
\newcommand{\coDP}{\ensuremath{\textnormal{coDP}}\xspace}
\newcommand{\ThetaTwo}{\ensuremath{\Theta_2^\textnormal{p}}\xspace}
\newcommand{\DeltaTwo}{\ensuremath{\Delta_2^\textnormal{p}}\xspace}
\newcommand{\ANDtwo}{\ensuremath{\textnormal{AND}_2}\xspace}
\newcommand{\ORtwo}{\ensuremath{\textnormal{OR}_2}\xspace}
\newcommand{\ANDomega}{\ensuremath{\textnormal{AND}_\omega}\xspace}
\newcommand{\ORomega}{\ensuremath{\textnormal{OR}_\omega}\xspace}
\newcommand{\coE}{\overline{E}}

\theoremstyle{plain}
\newtheorem{observation}[theorem]{Observation}
\title{Complexity of Stability}

\author{Fabian Frei}{Department of Computer Science, ETH Z\"{u}rich, Switzerland}{fabian.frei@inf.ethz.ch}{https://orcid.org/0000-0002-1368-3205}{}
\author{Edith Hemaspaandra}{Department of Computer Science, Rochester Institute of Technology, NY, USA}{eh@cs.rit.edu}{https://orcid.org/0000-0002-7115-626X}{Research done in part while on sabbatical at Heinrich-Heine-Universit\"at D\"usseldorf and  supported in part by NSF grant DUE-1819546 and a Renewed Research Stay grant from the Alexander von Humboldt Foundation.}
\author{J\"{o}rg Rothe}{Institut f\"{u}r Informatik, Heinrich-Heine-Universit\"{a}t D\"{u}sseldorf, Germany}{rothe@hhu.de}{https://orcid.org/0000-0002-0589-3616}{Research supported by DFG grants RO~1202/14-2 and RO~1202/21-1.}

\authorrunning{F. Frei, E. Hemaspaandra, and J. Rothe}
\Copyright{Fabian Frei, Edith Hemaspaandra, and J\"{o}rg Rothe}



\relatedversion{This paper is identical with the paper under the same name presented at the 31st International Symposium on Algorithms and Computation (ISAAC 2020)~\cite{fre-hem-rot:c:complexity-of-stability}, except for the inclusion of the appendices.}

\acknowledgements{We thank the anonymous referees for their careful reading of this paper and suggestions for improvement.}
\begin{document}
\maketitle
\begin{abstract}
Graph parameters such as the clique number, the chromatic number, and the
independence number are central in many areas, ranging from
computer networks to linguistics to computational neuroscience to
social networks.  In particular, the chromatic number of a graph
(i.e., the smallest number of colors needed to color all vertices such that no two
adjacent vertices are of the same color) can be applied in solving
practical tasks as diverse as pattern matching, scheduling jobs to
machines, allocating registers in compiler optimization, and even
solving Sudoku puzzles.  Typically, however, the underlying graphs are
subject to (often minor) changes.  To make these applications of graph
parameters robust, it is important to know which graphs are stable for
them in the sense that adding or deleting single edges or vertices
does not change them.  We initiate the study of stability of graphs
for such parameters in terms of their computational complexity.  We
show that, for various central graph parameters, the problem of
determining whether or not a given graph is stable is complete for
\ThetaTwo, a well-known complexity class in the second level of the
polynomial hierarchy, which is also known as ``parallel access to NP.''
\end{abstract}

\enlargethispage{-0.5\baselineskip}
\section{Introduction}
In this first section, we motivate our research topic, 
introduce the necessary notions and notation, 
and provide an overview of both the related work and our contribution.
\subsection{Motivation}
The following extends an example given by
Bollob\'as~\cite{bol:b:modern-graph-theory} (see also
\cite[Section~2.3.3]{jac:b:social-and-economic-networks}).
Consider a graph whose vertices represent the ISAAC-2020 attendees, 
and an edge
between any two vertices representing the wish of these two researchers
to attend each other's talks.  
The ISAAC-2020 organizers have elicited this information in advance, 
wishing to ensure that every participant can attend all 
the desired talks and also give their own
presentation.  
Therefore, they color this graph, 
where each color
represents a time slot 
(running parallel sessions within each time
slot).  
What is the smallest number of colors needed 
so that no two adjacent vertices have the same color, 
i.e., what is the chromatic
number of this graph?  
Suppose it is 12; so 12 time slots are enough
to make all the participants happy.  
Now, however, the ISAAC-2020 organizers receive messages 
from Professor Late and Professor Riser 
expressing their wish to attend each other's talks as well.  
Does this additional edge increase the chromatic number 
of the graph, requiring an additional time slot?  
Or does it always remain the same; that is, is this graph stable 
with respect to the chromatic number and adding edges?

Informally stated, a graph is 
\emph{stable with respect to some graph	parameter} 
(such as the chromatic number) if some type of small perturbation
of the graph 
(a local modification such as adding an edge or deleting a vertex) 
does not change the parameter.  
Other graph parameters we consider are the 
clique number, the independence number, and the vertex cover number.  
This notion of stability formalizes the robustness of graphs 
for these parameters, 
which is important in many applications.  
Typical applications of the chromatic number, for instance, include
coloring algorithms for complex networks such as social, economic,
biological, and information networks (see, e.g., Jackson's book on
social and economic networks~\cite{jac:b:social-and-economic-networks}
or Khor's work on applying graph coloring to biological
networks~\cite{kho:j:application-of-graph-coloring-to-biological-networks}).
In particular, social networks can be colored to find
roles~\cite{eve-bor:j:role-colouring-a-graph} or to study human
behavior in small controlled
groups~\cite{kea-sur-mon:j:experimental-study-of-coloring-problem-on-human-subject-networks,cha-gra-jam:c:network-coloring-game}.
In various applied areas of computer science, graph coloring has also
been used for register allocation in compiler
optimization~\cite{cha:j:register-allocation-spilling-via-graph-coloring},
pattern matching and pattern
mining~\cite{sun-tso-hok-fal-eli:j:two-heads-better-than-one},
and scheduling
tasks~\cite{lei:j:graph-coloring-algorithm-for-large-scheduling-problems}.
To ensure that these applications of graph parameters
are robust, graphs need to be stable for them with respect to certain
operations. We initiate a systematic study of stability of graphs
in terms of their computational complexity.

\subsection{Notions and Notation}
In this subsection, we define the core notions used in this paper and fix our notation. 
\subsubsection{Complexity Classes}
We begin with the relevant complexity classes. Besides \P, \NP, and \coNP, these are \DP, \coDP, and \ThetaTwo. 
The class \DP, introduced by Papadimitriou and Yannakakis~\cite{pap-yan:j:facets}, is the second level of the Boolean hierarchy over \NP; that is, $\DP=\NP\wedge \coNP=\{L_1\cap L_2\mid L_1\in\NP\wedge L_2\in \coNP\}$ is the set of all intersections of \NP languages with \coNP languages. Equivalently, it can be seen as the \emph{differences} of \NP languages, whence the name. 
An example of a trivially \DP-complete language is $\alternativetextsc{Sat-UnSat}=\Sat\times\UnSat$, where \UnSat is the set of all unsatisfiable \CNF-formulas. 
The complement class \coDP contains exactly the unions of \NP languages with \coNP languages. 

The class \ThetaTwo, whose name is due to Wagner~\cite{wag:j:bounded}, belongs to the second level of the polynomial hierarchy; it can be defined as $\ThetaTwo=\P^{\NP[\mathcal{O}(\log n)]}$, which is the class of problems that can be solved in polynomial time by an algorithm with access to an oracle that decides arbitrary instances for an \NP-complete problem -- with one instance per call and each such query taking constant time -- restricted to a logarithmic number of queries. (Without the last restriction, we would get the class $\DeltaTwo=\P^\NP$.)
Results due to Hemachandra~\cite[Thm.~4.10]{hem:j:sky} usefully characterize \ThetaTwo as $\P_\text{tt}^\text{p}$, the class of languages that are polynomial-time truth-table reducible to \NP. By definition, this is the same as $\P^\NP_\|$, the class of languages that are polynomial-time recognizable with \emph{unlimited parallel} access to an \NP oracle. 
\emph{Unlimited} means that an algorithm witnessing the membership of a problem in $\P^\NP_\|$ can query the oracle on as many instances of an \NP-complete problem as it wants -- which due the polynomial running-time means at most polynomially many -- while \emph{parallel} means that all queries need to be sent simultaneously. The characterization of \ThetaTwo as $\P^{\NP[\mathcal{O}(\log n)]}$, in contrast, allows the logarithmically many queries to be \emph{adaptive}; that is, they can be sent interactively, with one depending on the oracle's answers to the previous ones. Membership proofs for \ThetaTwo are usually easy; we will see a simple example of how to give one at the beginning of Section~\ref{sec:stability-for-colorability}. 

Note that the definitions immediately yield the inclusions $\NP\cup \coNP\subseteq \DP\subseteq \ThetaTwo\subseteq\DeltaTwo$. 

\subsubsection{Graphs and Graph Numbers}
Throughout this paper graphs are simple. Let $\mathcal{G}$ be the set of all (simple) graphs and $\N$ the set of natural numbers including zero. For any set $M$, we denote its \emph{cardinality} or \emph{size} by $\card{M}$. 
A map $\graphnumber\colon \mathcal{G}\to\N$ is called a \emph{graph number}. In this paper, we examine the prominent graph numbers $\alpha$, $\beta$, $\chi$, and $\omega$, which give the size of a maximum independent set, the size of a minimum vertex cover, the size of a minimum coloring (i.e., the minimum number of colors allowing for a proper vertex coloring), and the size a maximum clique, respectively. 

Let $V$, $E$, and $\coE$ be the functions that map a graph $G$ to its vertex set $V(G)$, its edge set $E(G)$, and its set of \emph{nonedges} $\coE(G)=\{\{u,v\}\mid u,v\in V(G)\wedge u\neq v\}\setminus E(G)$, respectively. 

Let $G$ and $H$ be graphs. 
We denote by $G\cup H$ the disjoint union and by $G+H$ the \emph{join}, which is $G\cup H$ with all \emph{join edges} -- i.e., the edges 
$\{v,w\}\in V(G)\times V(H)$ -- added to it.\footnote{We adopt the notation $G+H$ for the join from Harary's classical textbook on graph theory~\cite[p. 21]{har:graph-theory}.}

For $v\in V(G)$, $e\in E(G)$, and $e'\in \coE(G)$, we denote by $G-v$, $G-e$, and $G+e'$ the graphs that result from $G$ by deleting $v$, deleting $e$, and adding $e'$, respectively. 

For any $k\in \N$, we denote by $I_k$ 
and $K_k$ the empty (i.e., edgeless) and complete graph on $k$ vertices,  respectively. The graph $I_0=K_0$ without any vertices is called the \emph{null graph}. 
A vertex $v$ is \emph{universal} with respect to a graph $G$ if it is adjacent to all vertices $V(G)\setminus\{v\}$. 

\subsubsection{Stability}
Let $G$ be a graph. An edge $e\in E(G)$ is called \emph{stable} with respect to a graph number $\graphnumber$ (or $\graphnumber$-\emph{stable}, for short) if $\graphnumber(G)=\graphnumber(G-e)$, that is, deleting $e$ leaves $\graphnumber$ unchanged. 
Otherwise (that is, if the deletion of $e$ does change $\graphnumber$), $e$ is called $\graphnumber$-\emph{critical}. 
For a vertex $v\in V(G)$ instead of an edge $e\in E(G)$, stability and criticality are defined in the same way.  

A graph is called $\graphnumber$-stable if all of its edges are $\graphnumber$-stable. A graph whose vertices -- rather than edges -- are all $\graphnumber$-stable is called $\graphnumber$-vertex-stable. The notions of $\graphnumber$-criticality and $\graphnumber$-vertex-criticality are defined analogously. 
Note that each edge and vertex is either stable or critical, whereas a graph might be neither. 
An unspecified $\graphnumber$ defaults to the chromatic number $\chi$. \pagebreak

A traditional term for stability with respect to adding edges and vertices -- rather than deleting them -- is unfrozenness.\footnote{The notion of instance parts being either frozen or unfrozen has originally been introduced to the field of computational complexity in analogy to the physical process of  freezing~\cite{mon-zec-:j:mechanics-of-random-sat,mon-zec-etal:j:phase-transition}. 
The sudden shift from \P to \NP-hardness that can be observed when transitioning from \ETwoSat to \EThreeSat by allowing a larger and larger percentage of clauses of length 3 rather than 2, for example, mimics the phase transition from liquid to solid, with the former granting much higher degrees of freedom to the substance's constituents than the latter. 
Based on this general intuition, Beacham and Culberson~\cite{bea-cul:j:complexity-unfrozen} then more formally defined the notion of unfrozenness with regard to an arbitrary graph property that is downward monotone (meaning that a graph keeps the property when edges are deleted); they call a graph unfrozen if it also keeps the property when an arbitrary new edge is added. 
We naturally extend this notion to arbitrary graph numbers, which are not necessarily monotone.} 
Specifically, a nonedge $e\in\coE(G)$ is called \emph{unfrozen} if adding it to the graph $G$ leaves $\chi$ unchanged, and \emph{frozen} otherwise.
All of these notions extend naturally to vertices (where we can freely choose to which existing vertices a new vertex is adjacent, implying an exponential number of possibilities), to entire graphs, and to any graph number $\graphnumber$, as just seen for stability and criticality. 

We call a graph \emph{two-way stable} if it is both stable and unfrozen, everything with respect to the chromatic number and deleting an edge as the default choice.  Again, we have the analogous set of notions with respect to vertices and any graph number $\graphnumber$. 

Prefixing a natural number $k\in\N$ to any of these notions additionally requires the respective graph number to be exactly $k$. For example, a graph $G$ is \emph{$k$-critical} if and only if $\chi(G)=k$ and $\chi(G-e)\neq k$ for every $e\in E(G)$. 

The notion of stability can be naturally applied to Boolean formulas as well. We call a formula $\Phi$ in conjunctive normal form \emph{stable} if deleting an arbitrary clause $C$ does not change its satisfiability status -- that is, if it either is satisfiable (and of course stays so upon deletion of a clause) or if it and all its 1-clause-deleted subformulas $\Phi-C$ are unsatisfiable. 

\subsubsection{Languages} 
We denote by \CNF the set of formulas in conjunctive normal form and by \ThreeCNF, \FourCNF, and \SixCNF the set of \CNF-formulas with \emph{exactly} 3, 4, and 6 \emph{distinct} literals per clause, respectively.\footnote{In the literature, these set names are often prefixed by an \alternativetextsc{E}, emphasizing the exactness. This is notably not the case for a paper by Cai and Meyer~\cite{cai-mey:j:dp} that contains a construction crucially relying on this restriction. We will build upon this construction later on and are thus bound to the same constraint.} 
The sets \Sat and \ThreeSat contain the satisfiable, \UnSat and \ThreeUnSat the unsatisfiable formulas from \CNF and \ThreeCNF, respectively. 
Let $\StableUnSat=\{\Phi\in\UnSat\mid (\Phi-C)\in\UnSat\text{ for every clause }C\text{ of }\Phi\}$ be the set of stably unsatisfiable formulas. 
The set $\StableCNF=\Sat\cup\StableUnSat$ consists of the stable \CNF-formulas. 
Intersecting with \ThreeCNF yields the classes \StableThreeUnSat and \StableThreeCNF and so on. 

Let \Stability be the set of stable graphs and \Unfrozenness the set of unfrozen graphs, both with respect to the default graph number $\chi$. The set of two-way stable graphs is $\TwoWayStability=\Stability\cap\Unfrozenness$. 
Once more, these definitions extend naturally. For example, $4\text{-}\VertexStability$ is the set of (with respect to the default $\chi$) $4$-vertex-stable graphs and $\beta$-\TwoWayStability consists of the graphs for which the vertex-cover number $\beta$ remains unchanged upon deletion or addition of an edge. 

\subsubsection{AND Functions and OR Functions}
Following Chang and Kadin~\cite{cha-kad:j:closer}, we say that a language $L\subseteq\Sigma^\ast$ has \ANDtwo if there is a polynomial-time computable function $f\colon \Sigma^\ast\times \Sigma^\ast\to \Sigma^\ast$ such that 
for all $x_1,x_2\in\Sigma^\ast$, we have $x_1\in L\wedge x_2\in L\,\Longleftrightarrow\,f(x_1,x_2)\in L$. If this is the case, we call $f$ an \ANDtwo function for $L$. 
If there even is a polynomial-time computable function $f: \bigcup_{k=0}^\infty(\Sigma^\ast)^k\to\Sigma^\ast$ such that for every 
$k\in\N$ and for all $x_1,\ldots,x_k\in\Sigma^\ast$ we have $x_1\in L\wedge\cdots \wedge x_k\in L\,\Longleftrightarrow\,f(x_1,\ldots,x_k)\in L$, 
then we say that $L$ has \ANDomega. 
Replacing $\wedge$ with $\vee$, we get the analogous notions \ORtwo and \ORomega. 
Note that a language has \ANDtwo if and only if its complement has \ORtwo, with the analogous statement holding for \ANDomega and \ORomega. 

\subsection{Related Work}
Many interesting problems are suspected to be complete for either \DP or \ThetaTwo. While membership is usually trivial in these cases, matching lower bounds are rare and hard to prove. 
For example, Woeginger~\cite{woe:hedonic-coalition-formation} observes that determining whether a graph has a wonderfully stable partition is in \ThetaTwo, and leaves it as an open problem to settle the exact complexity. 
Wagner, who introduced the class name \ThetaTwo~\cite{wag:j:more-on-bh}, provided a number of hardness results for variants of standard problems such as Satisfiability, Clique and Colorability, which are designed to be complete for \DP or \ThetaTwo. 
For example, he proves the \DP-completeness of  $\alternativetextsc{ExactColorability}=\{(G,k)\in\mathcal{G}\times\N\mid \chi(G)=k\}$~\cite[Thm.~6.3.1 with $k=1$]{wag:j:more-on-bh} and the \ThetaTwo-completeness of $\alternativetextsc{OddVertexCover}=\{G\in\mathcal{G}\mid \beta(G)\text{ is odd}\}$~\cite[Thm.~6.1.2]{wag:j:more-on-bh}.\footnote{Note that Wagner originally derived his results with respect to the more restricted form of polynomial-time reducibility via Boolean formulas, indicated by the \emph{bf} in 
the class 
name. He later proved the resulting notions to be equivalent, however; that is, we have $\P_\text{bf}^\text{p}=\ThetaTwo$~\cite{wag:j:bounded}.} 
He obtains the analogous results for Colorability, Clique~\cite[Thm.~6.3]{wag:j:more-on-bh}, Independent Set instead of Vertex Cover~\cite[Thm.~6.4]{wag:j:more-on-bh} and points out~\cite[second-to-last paragraph]{wag:j:more-on-bh} that his proof techniques also yield the \ThetaTwo-completeness of the equality version of all of these problems -- for example, $\alternativetextsc{EqualVertexCover}=\{(G,H)\in\mathcal{G}^2\mid \beta(G)=\beta(H)\}$. 
The same holds true for the comparison versions such as $\alternativetextsc{CompareVertexCover}=\{(G,H)\in\mathcal{G}^2\mid \beta(G)\le\beta(H)\}$.\footnote{Spakowski and Vogel explicitly proved the \ThetaTwo-completeness of $\alternativetextsc{CompareVertexCover}$~\cite[Thm.~12]{spa-vog:c:theta-two-classic}, $\alternativetextsc{CompareClique}$ and $\alternativetextsc{CompareIndependentSet}$~\cite[Thm.~13]{spa-vog:c:theta-two-classic}. 
For other cases, see Appendix~\ref{app:CompareColorability}.
}
The \DP-completeness of \alternativetextsc{ExactColorability} has been extended to the subproblem of recognizing graphs with chromatic number $4$~\cite{rot:j:exact-four-colorability}.
Furthermore, a few election problems have been proved to be \ThetaTwo-complete by Hemaspaandra et al.~\cite{hem-hem-rot:j:dodgson, hem-hem-rot:j:online}, by Rothe et al.~\cite{rot-spa-vog:j:young}, and Hemaspaandra et al.~\cite{hem-spa-vog:j:kemeny}. 

\enlargethispage{1.5\baselineskip}
In general, establishing lower bounds proved to be difficult for many natural \DP-complete and particularly \ThetaTwo-complete problems. Consequently, hardness results remained rather rare in the area of criticality and stability, despite the great attention that these natural notions have garnered from graph theorists ever since the seminal paper by Dirac~\cite{dir:theorems-abstract-graphs} from 1952; see for example the classical textbooks by Harary~\cite[chapters 10 and 12]{har:graph-theory} and Bollob\'as~\cite[chapter IV]{bol:b:modern-graph-theory} -- the latter having a precursor dedicated exclusively to extremal graph theory~\cite[chapters I and V]{bol:b:extremal-graph-theory} -- and countless papers over the decades, of which we cite some selected examples from early to recent ones~\cite{
erd-gal:j:alpha-criticality, har-tho:j:anticritical-graphs, bau-har-etal:j:domination-alteration-sets, wes:criticality-graph-theory, gun-har-ral:j:is-stability, hay-bri-etal:j:insensitive-domination, des-hay-hen:j:total-domination-critical-and-stable, hen-krz:j:total-domination-stability, che-jin:j:critical-graph-structures}.
A pioneering complexity result by Papadimitriou and Wolfe~\cite[Thm.~1]{pap-wol:j:facets} establishes the \DP-completeness of \MinimalUnSat. (They call a formula minimally unsatisfiable if deleting an arbitrary clause renders it satisfiable, that is, if it is critical.) 
They also proved 
that determining, given a graph $G$ and a $k\in\N$, whether $G$ is $k$-$\omega$-vertex-critical is a
\DP-complete problem~\cite[Thm.~4]{pap-wol:j:facets}. 
Later, Cai and Meyer~\cite{cai-mey:j:dp} showed the \DP-completeness of  $k$-\VertexCriticality (which they call \alternativetextsc{Minimal}-$k$-\alternativetextsc{Uncolorability}) for all $k\ge 3$.  
Burjons et al.~\cite{bur-fre-etal:c:neighborly-help} recently extended this result to the more difficult case of edge deletion, showing that $k$-\Criticality is \DP-complete for all 
$k\ge 3$~\cite[Thm.~8]{bur-fre-etal:c:neighborly-help}. 
They also provided the first \ThetaTwo-hardness result for a criticality problem, namely for $\beta$-\VertexCriticality~\cite[Thm.~15]{bur-fre-etal:c:neighborly-help}. 
Note the drop in difficulty down to \DP
when fixing the graph number. This emerges as a general pattern, as evidenced by our results outlined 
in the contribution section below. 

Stability, in contrast to criticality,
has been sorely neglected by the computational complexity community, which is surprising in light of its apparent practical relevance -- for example in the design of infrastructure, where stability is a most desirable property. 
A very small exception to this are
Beacham and Culberson~\cite{bea-cul:j:complexity-unfrozen}, who proved a comparably easy variant of Unfrozenness, namely $\{(G,k)\mid \chi(G)\le k\text{ and }G\text{ is unfrozen}\}$, to be $\NP$-complete.
\vspace{-0.3\baselineskip}
\subsection{Contribution}
We choose four of the most prominent graph problems -- Colorability, Vertex Cover, Independent Set, and Clique -- to analyze the complexity of stability.   
We prove all of them to be \ThetaTwo-complete for the default case of edge deletion. For unfrozenness -- that is, stability with respect to edge addition -- we prove the same, with the one exception of Colorability. For this problem, we prove that the existence of a construction with a few simple properties would be sufficient to prove \ThetaTwo-completeness. Finally, we introduce the notion of two-way stability -- stability with respect to both deleting and adding edges -- and prove again \ThetaTwo-completeness for all four problems. Table~\ref{tab:OverviewTable} provides an overview of these results, showcasing surprising contrasts between some of the problems.  

We also derive several other useful results with broad appeal on their own, among these being the \coDP-completeness of \StableThreeCNF [Thm.~\ref{thm:StableThreeCNF}], the \DP-completeness of $k$-\Stability and $k$-\VertexStability for all $k\ge 4$ [Thm.~\ref{thm:k-stability}], general criteria for proving \DP-hardness [Lems.~\ref{lem:equal} and~\ref{lem:compare}], and finally constructions such as the edge-stabilizing gadget [Lem.~\ref{lem:StabilizingGadget}] that yields an \ANDomega function for \Stability [Cor.~\ref{cor:and-for-stability}] and has potential applications in various contexts such as reoptimization and general graph theory. 

\begin{table}[htbp]
\centering
\caption{An overview of our results regarding the complexity of different stability problems. See Section~\ref{sec:CliqueAndIS} for the results on Clique and Independent Set; almost all of them follow in analogy to the ones for Vertex Cover, with \alphaVertexStability and \omegaVertexStability being the exception. 
}
\label{tab:OverviewTable}
\setlength{\tabcolsep}{4pt}
\begin{tabular}{p{0.23\textwidth}cccccc}
\toprule
\multirow{3}{*}{\shortstack[l]{\ \\\textit{With respect to this}\\\textit{base problem and}\\\textit{graph number:}}}
&\multicolumn{6}{c}{\hfill}\\
& \multicolumn{2}{c}{Stability} 
& \multicolumn{2}{c}{Unfrozenness} 
& \multicolumn{2}{c}{Two-Way Stability} 
\\ 
\cmidrule(lr){2-3}\cmidrule(lr){4-5} \cmidrule(lr){6-7} 
 & Edge 
& Vertex& Edge 
& Vertex& Edge 
& Vertex 
\\  
\cmidrule(lr){2-7}
&[Thm.\,\ref{thm:beta-stability}]
&[Thm.\,\ref{thm:beta-vertex-stability}]
			
&[Thm.\,\ref{thm:beta-unfrozenness}]
&[Thm.\,\ref{thm:v-unfrozenness}]
&[Thm.\,\ref{thm:beta-twowaystability}]
&[Thm.\,\ref{thm:v-twowaystability}]\\
\textit{Vertex Cover, }$\beta$	
& $\ThetaTwo\text{-compl.}$
& $\P$   
			
& $\ThetaTwo\text{-compl.}$
&\P
& $\ThetaTwo\text{-compl.}$
&\P\\\cmidrule(lr){2-7}
			
\!\!\begin{tabular}{l}\textit{Independent Set,} $\alpha$\\\textit{and Clique, $\omega$
}\end{tabular}\!\!
& \!\!\begin{tabular}{l}$\ThetaTwo\text{-compl.}$\end{tabular}\!\!
& \!\!\begin{tabular}{l}$\ThetaTwo\text{-compl.}$\end{tabular}\!\!
& \!\!\begin{tabular}{l}$\ThetaTwo\text{-compl.}$\end{tabular}\!\!
& \!\!\begin{tabular}{l}$\P$\end{tabular}\!\!
& \!\!\begin{tabular}{l}$\ThetaTwo\text{-compl.}$\end{tabular}\!\!
& \!\!\begin{tabular}{l}$\P$\end{tabular}\!\!\\\cmidrule(lr){2-7}
			
\textit{Colorability, }$\chi$
& $\ThetaTwo\text{-compl.}$
& $\ThetaTwo\text{-compl.}$   	
& ?
&\P   	
& $\ThetaTwo\text{-compl.}$  
&\P\\
			
&[Thm.\,\ref{thm:stability}]
&[Thm.\,\ref{thm:v-stability}]			
&[Thm.\,\ref{thm:unfrozenness}]
&[Thm.\,\ref{thm:v-unfrozenness}]			
&[Thm.\,\ref{thm:twowaystability}]
&[Thm.\,\ref{thm:v-twowaystability}]\\								
\bottomrule
\end{tabular}

\end{table}

\vspace{-0.4\baselineskip}
\enlargethispage{\baselineskip}
\section{Basic Observations}
We begin with a few very basic and useful observations that will be used implicitly and, where appropriate, explicitly throughout the paper. The proofs are given in Appendix~\ref{app:observations}. 
\begin{observation}\label{obs:differsbyatmostone}
The deletion of an edge or of a vertex either decreases the chromatic number by exactly one or leaves it unchanged. 
\end{observation}
\begin{observation}\label{obs:CriticalityFromEdges}
Let $e=\{u,v\}$ be a critical edge. Then $u$ and $v$ are critical as well. 
\end{observation}
\begin{observation}\label{obs:StabilityFromVertices}
Let $v$ be a stable vertex. Then all edges incident to $v$ are stable.  
\end{observation}
\begin{observation}\label{obs:CriticalVertexCharacterization}
Let $G$ be a graph. A vertex $v\in V(G)$ is critical if and only if there is an optimal coloring of $G$ that assigns $v$ a color with which no other vertex is colored. 
\end{observation}

\section{Connections between Clique, Vertex Cover, and Independent Set}
\label{sec:CliqueAndIS}
As is to be expected, the three problems of Clique, Vertex Cover, and Independent Set are so closely related that almost all stability results for one of them carry over to the other two in a straightforward way. We state the connections in Proposition~\ref{pro:CliqueAndIS}, proved in Appendix~\ref{app:CliqueAndIS}. 
\begin{proposition}\label{pro:CliqueAndIS}
Let $\coG$ denote the complement graph of $G$. We have the following equalities.
\begin{enumerate}
\item 
$\beta\text{-}\Stability=\alpha\text{-}\Stability=\{\coG\mid G\in \omega\text{-}\Unfrozenness\}$.
\item 
$\beta\text{-}\Unfrozenness=\alpha\text{-}\Unfrozenness=\{\coG\mid G\in \omega\text{-}\Stability\}$.
\item 
$\beta\text{-}\TwoWayStability=\alpha\text{-}\TwoWayStability=\{\coG\mid G\in \omega\text{-}\TwoWayStability\}$.
\item $\beta\text{-}\VertexStability=\{I_n\mid n\in\N\}$.
\item $\alpha\text{-}\VertexStability=\{\coG \mid G\in\omega\text{-}\VertexStability\}$. 
\item
$\beta\text{-}\VertexUnfrozenness=\beta\text{-}\VertexTwoWayStability=\{K_0\}$.
\item $\alpha\text{-}\VertexUnfrozenness=\alpha\text{-}\VertexTwoWayStability=\\
\omega\text{-}\VertexUnfrozenness=\omega\text{-}\VertexTwoWayStability=\emptyset$.
\end{enumerate} 
\end{proposition}

An interesting inversion in this pattern occurs for the vertex deletion case. Here, switching from $\beta$ to $\alpha$ or $\omega$ in fact flips the stability problem to the criticality version and vice versa. 

\begin{proposition}\label{pro:StabilityAndCriticality}
We have the following equalities. 
\begin{enumerate}
\item $\betaVertexStability=\alpha\text{-}\VertexCriticality=\{\coG\mid G\in \omega\text{-}\VertexCriticality\}$.
\item $\beta\text{-}\VertexCriticality=\alpha\text{-}\VertexStability=\{\coG\mid G\in \omega\text{-}\VertexStability\}$.
\end{enumerate}
\end{proposition}

Proposition~\ref{pro:StabilityAndCriticality} is proved in Appendix~\ref{app:StabilityAndCriticality}. 
Using it, we directly obtain from the \ThetaTwo-hardness of $\beta$-\VertexCriticality~\cite{bur-fre-etal:c:neighborly-help} the same for $\alpha$-\VertexStability and, by complementing the graphs, $\omega$-\VertexStability. 
Unfortunately, $\beta$-\VertexCriticality is the only problem to yield any nontrivial result via the connection between stability and criticality. 

We now turn our attention to the remaining stability problems, for which the hardness proofs will require substantially more effort. 

\section{Stability and Vertex-Stability for Colorability}\label{sec:stability-for-colorability}
We will prove \ThetaTwo-completeness for both \Stability and \VertexStability. 
\begin{theorem}\label{thm:stability}
Determining whether a graph is stable is \ThetaTwo-complete.
\end{theorem}

\begin{theorem}\label{thm:v-stability}
Determining whether a graph is vertex-stable
is \ThetaTwo-complete. 
\end{theorem}

As is typical, the upper bounds are immediate: 
We can determine the chromatic number of a graph and all its 1-vertex-deleted and
1-edge-deleted subgraphs with a polynomial number of parallel queries to an oracle for the standard, \NP-complete colorability 
problem $\{(G,k)\in\mathcal{G}\times\N\mid \chi(G)\le k\}$. Specifically, the queries $(G,k)$, $(G-e,k)$, and $(G-v,k)$ for every $e\in E(G)$, every $v\in V(G)$, and every $k\in\{0,\ldots, \card{V(G)}\}$ suffice to determine
whether $G$ is stable and whether it
is vertex-stable. 
To prove the matching lower bounds,
we first note that the lower bound for
Theorem~\ref{thm:v-stability} implies the 
lower bound for Theorem~\ref{thm:stability}. 

\begin{lemma}\label{lem:v-stability-to-stability}
\VertexStability polynomial-time many-one reduces to \Stability.
\end{lemma}
It can be shown that mapping a graph $G$ to its self-join $G+G$ provides the required reduction.  Due to the space restrictions, the proof of Lemma~\ref{lem:v-stability-to-stability} is deferred to Appendix~\ref{app:v-stability-to-stability}. 

It remains to establish the lower bound of Theorem~\ref{thm:v-stability}, that is, 
to prove that determining whether a graph is vertex-stable
is \ThetaTwo-hard. 
Proving \ThetaTwo-hardness is not easy.
However, we will now argue that it suffices to show
that \VertexStability is \coDP-hard. 

Chang and Kadin~\cite[Thm.~7.2]{cha-kad:j:boolean-connectives}
show that a problem is \ThetaTwo-hard if
it is \DP-hard and has \ORomega.
Observing that \ThetaTwo is closed under complement, we obtain the following corollary.

\begin{corollary}\label{cor:chang-kadin}
If a \coDP-hard problem has \ANDomega, 
then it is \ThetaTwo-hard.
\end{corollary}

We note that the join is an \ANDomega function for \VertexStability; see Appendix~\ref{app:JoinIsANDomega}. 
Now, Theorem~\ref{thm:v-stability} follows from Corollary~\ref{cor:chang-kadin}
and the \coDP-hardness of \VertexStability.
\begin{theorem}\label{thm:JoinIsANDomega}
The join is an \ANDomega function for \VertexStability and \Unfrozenness.
\end{theorem}
\begin{lemma}\label{lem:v-stability-coDP}
Determining whether a graph is vertex-stable is \coDP-hard.
\end{lemma}

To prove Lemma~\ref{lem:v-stability-coDP}, we show in Theorem~\ref{thm:StableThreeCNF} that $\StableThreeCNF=\ThreeSat\cup\StableThreeUnSat$ is \coDP-complete and then reduce it to \VertexStability in Theorem~\ref{thm:cai-meyer}. 
We will use twice 
the following lemma, whose straightforward proof is deferred to Appendix~\ref{app:SatToThreeSat}. 

\begin{lemma}\label{lem:SatToThreeSat}
	There is a polynomial-time many-one reduction from \Sat to \ThreeSat converting a \CNF-formula $\Phi$ into a \ThreeCNF-formula $\Psi$ such that 
	$\Phi$ is stable if and only if $\Psi$ is stable. 
\end{lemma}

\begin{theorem}\label{thm:StableThreeCNF}
	\StableThreeCNF is \coDP-complete.
\end{theorem}

\begin{proof}
It is immediate that \StableThreeCNF is in \coDP.
To show \coDP-hardness, we will show that
\StableThreeCNF is \coNP-hard, \NP-hard, and has \ORtwo. 
The \coDP-hardness then follows by applying an observation by Chang and Kadin~\cite[Lem. 5]{cha-kad:j:boolean-connectives} -- a set is \DP-hard if it is \NP-hard, \coNP-hard, and has an \ANDtwo function -- to the complement language. 
\begin{description}
\item[coNP-hardness.]
It is easy to see that the function $f\colon \Phi\mapsto \Phi \wedge (x\vee y \vee z) \wedge (x\vee y \vee \overline{z}) \wedge (x\vee \overline{y} \vee z) \wedge (x\vee \overline{y} \vee \overline{z}) \wedge (\overline{x}\vee y \vee z) \wedge (\overline{x}\vee y \vee \overline{z}) \wedge (\overline{x}\vee \overline{y} \vee z) \wedge (\overline{x}\vee \overline{y} \vee \overline{z}) $, where $x$, $y$, and $z$ are fresh variables 
not occurring in $\Phi$, reduces
\ThreeUnSat to \StableThreeCNF. 

\item[NP-hardness.]
We give a reduction from \ThreeSat to \StableFourCNF;  composing it with the reduction from Lemma~\ref{lem:SatToThreeSat} yields the desired reduction to \StableThreeCNF. 
Given a \ThreeCNF-formula $\Phi = C_1 \wedge \cdots \wedge C_m$
over $X=\{x_1,\ldots,x_n\}$, map it to
the \FourCNF-formula $\Psi = (C_1 \vee y) \wedge (C_1' \vee y') \wedge (C_1'' \vee y'') \wedge \cdots \wedge (C_m\vee y) \wedge (C_m'\vee y') \wedge (C_m''\vee y'') \wedge
(\overline{y} \vee \overline{y}' \vee \overline{y}'')$,
where the clauses $C'_i$ and $C''_i$ are just like the clauses $C_i$ but with a new copy of variables $X'=\{x'_1,\ldots,x'_n\}$ and $X''=\{x''_1,\ldots,x''_n\}$ instead of $X$, respectively, 
and $y$, $y'$, and $y''$ being three fresh variables as well.
Deleting the clause $(\overline{y} \vee \overline{y}' \vee \overline{y}'')$ renders $\Psi$
trivially satisfiable; 
any assignment that sets $y$, $y'$ and $y''$ to 1 will do.
Thus $\Psi$ is stable if and only if it is satisfiable.
It remains to prove the equisatisfiability of  
$\Phi$ and $\Psi$. 

First assume that $\Phi$ has a satisfying assignment $\assignment\colon X\to \{0,1\}$. Then $\Psi$ is satisfied by any assignment $\assignmenttwo$ with 
$\assignmenttwo(x_i')=\assignmenttwo(x_i'')=\assignment(x_i)$ for $i\in\{1,\ldots,n\}$ and 
$\assignmenttwo(y)=0$.
Now assume that $\Psi$ has a satisfying assignment $\assignmenttwo$.
Then $\Phi$ is satisfied by $\assignment: x_i\mapsto \assignmenttwo(x_i)$ if $\assignmenttwo(y)=0$, by $\assignment': x_i\mapsto \assignmenttwo(x_i')$ if $\assignmenttwo(y')=0$, and by $\assignment'': x_i\mapsto \assignmenttwo(x_i'')$  if $\assignmenttwo(y'')=0$.

\item[OR$_2$.] 
In their proof of \DP-completeness, Papadimitriou and Wolfe~\cite[Lem. 3 plus corollary]{pap-wol:j:facets} implicitly gave a simple \ANDtwo function for both \MinimalUnSat and \MinimalThreeUnSat (the sets of unsatisfiable
formulas that become satisfiable after deleting any clause). 
We make use of the same construction and defer the full proof to Appendix~\ref{app:or}. 
\end{description}
This concludes the proof that \StableThreeCNF is \coDP-complete. 
\end{proof}

All that is left to do is to reduce \StableThreeCNF to \VertexStability.
First, we consider the known reduction from 
\MinimalThreeUnSat to \VertexMinimalThreeUnCol 
by Cai and Meyer~\cite{cai-mey:j:dp}. 
It maps a formula $\Phi$ with $m$ clauses $C_1,\ldots,C_m$ 
to a graph $G_\Phi$, 
whose vertex set includes, among others, a vertex called $v_\textnormal{s}$ and, 
for every $i\in\{1,\ldots,m\}$, a vertex $t_{i1}$;  
see Figure~\ref{fig:caimeyertotal} in Appendix~\ref{app:caimeyertotal} for an example of the full construction, combining the single steps described in the original paper~\cite{cai-mey:j:dp}.
It comes as no surprise that this reduction 
does not work for us 
since, for example, $G_\Phi - v_\textnormal{s}$ 
is always 3-colorable, and thus 
$G_\Phi$ is never stable if $\Phi$ is not satisfiable.
However, careful checking reveals the following property of $G_\Phi$. 
\begin{lemma}\label{lem:cai-meyer}
A \ThreeCNF-formula $\Phi$ is not stable  
if and only if
$\chi(G_\Phi) > \chi(G_\Phi - t_{i1})$ for at least one $i\in\{1,\ldots,m\}$.
\end{lemma}
The proof of Lemma~\ref{lem:cai-meyer} is deferred to Appendix~\ref{app:cai-meyer}. 
What we need now is a way to enhance the construction such that the deletion of a vertex other than $t_{11},\ldots,t_{m1}$, for example $v_\textnormal{s}$, does not decrease the chromatic number. We achieve this by the following lemma. 

\begin{lemma}\label{lem:doubling}
Let $G$ be a graph and $v\in V(G)$. Let $\widehat{G}$ be the graph that results from replicating $v$; that is, $V(\widehat{G}) = V(G) \cup \{v'\}$ and
$E(\widehat{G}) = E(G) \cup \{ \{v',w\} \ | \ \{v,w\} \in E(G)\}$. 
Then $\chi(G) = \chi(\widehat{G}) = \chi(\widehat{G} - v) = \chi(\widehat{G} - v')$.
\end{lemma}

\begin{proof}
The only nontrivial part is to show that $\chi(\widehat{G}) \leq \chi(G)$.
To see this, we start with an arbitrary optimal valid vertex coloring of $G$ and then color $v'$ with the same color
as $v$.
\end{proof}

Lemma~\ref{lem:doubling} is simple and yet very powerful in our context. It allows us to select a set of vertices whose removal will not influence the chromatic number, and thus will not influence whether or not the graph is vertex-stable.
We can use this to obtain the desired reduction.
\begin{theorem} \label{thm:cai-meyer}
\StableThreeCNF polynomial-time many-one reduces to \VertexStability.
\end{theorem}
\begin{proof}
Given a \ThreeCNF-formula $\Phi$, map it to $r({G}_\Phi)$, where $G_\Phi$ is the graph from the reduction by Cai and Meyer~\cite{cai-mey:j:dp} and $r$ denotes the replication of all vertices other than $t_{11},\dots,t_{m1}$. 

If $\Phi$ is not in \StableThreeCNF, then we have $\chi(G_\Phi) > \chi(G_\Phi - t_{i1})$ for some $i\in\{1,\dots,m\}$ by Lemma~\ref{lem:cai-meyer}. 
Furthermore, a repeated application of Lemma~\ref{lem:doubling} yields $\chi(r({G}_\Phi)) = \chi(G_\Phi)$ and
$\chi(r(G_\Phi) - t_{i1}) = \chi(r(G_\Phi - t_{i1})) =
\chi(G_\Phi - t_{i1})$.
Thus $r(G_\Phi)$ is not vertex-stable.
For the converse, suppose that $r(G_\Phi)$ is not vertex-stable.
Let $v\in V(r(G_\Phi))$ be a vertex such that $\chi(r(G_\Phi)) > \chi(r(G_\Phi) - v)$.
From Lemma~\ref{lem:doubling}, we can see that $v = t_{i1}$ for some $i\in\{1,\dots,m\}$. 
By Lemma~\ref{lem:cai-meyer}, this implies 
that $\Phi$ is not stable.
\end{proof}

This completes the proof of Theorem~\ref{thm:v-stability} -- stating that \VertexStability is \ThetaTwo-complete -- which in turn implies Theorem~\ref{thm:stability}, the \ThetaTwo-completeness of \Stability, by Lemma~\ref{lem:v-stability-to-stability}.
Now we briefly turn to some \DP-complete problems. 
Recall that by prefixing a number $k$ to the name of a stability property we additionally require the graph number to be exactly $k$. 

\begin{theorem}\label{thm:k-stability}
The problems $k$-\Stability and $k$-\VertexStability are \NP-complete for $k = 3$ and \DP-complete for $k\ge 4$. 
\end{theorem}
\begin{proof}
The membership proofs are immediate. For the lower bound we use that \alternativetextsc{Exact}-$k$-\alternativetextsc{Colorability} (the class of all graphs whose chromatic number is not merely at most, but exactly $k$) is $\NP$-complete for $k=3$ and \DP-complete for $k\ge 4$; see~\cite{rot:j:exact-four-colorability}. 
It suffices to check that mapping $G$ to $G\cup G$ reduces \alternativetextsc{Exact}-$k$-\alternativetextsc{Colorability} to $k$-\Stability and $k$-\VertexStability. 
Indeed, for any two graphs $H$ and $H'$, we have $\chi(H\cup H')= \max\{\chi(H),\chi(H')\}$, implying that $G\cup G$ is stable and vertex-stable with $\chi(G)=\chi(G\cup G)$.
\end{proof}

In the previous proof, we used the disjoint union of a graph with itself to render it stable without changing its chromatic number. 
Using a far more complicated construction, we can also ensure the stability of an arbitrary set of edges of a graph while keeping track of how exactly this changes the chromatic number. 
We state this result in the following theorem. 

\begin{lemma}\label{lem:StabilizingGadget}
There is a polynomial-time algorithm that, given any graph $G$ plus a nonempty subset $S\subseteq E(G)$ of its edges, adds a fixed gadget to the graph and then substitutes for every $e\in S$ some gadget that depends on $G$ and $e$, yielding a graph $\widehat{G}$ with the following properties:
\begin{enumerate}
\item 
$\chi(\widehat{G})=\chi(G)+2$.
\item 
All edges in $E(\widehat{G})\setminus (E(G)\setminus S)$ are stable. 
\item Each one of the remaining edges in $E(G)\setminus S$ is stable in $\widehat{G}$ exactly if it is stable in $G$.
\end{enumerate}
\end{lemma} 
Owing to space constraints, the proof of Lemma~\ref{lem:StabilizingGadget} is deferred to Appendix~\ref{app:StabilizingGadget}. 
However, we can at least provide a brief sketch of the construction for the case where only one edge is stabilized, that is $S=\{e\}$, omitting the verification of the properties. Join a new edge 
$\{w_1',w_2'\}$ 
to the given $G$, remove $e$, join one of its endpoints to $G_e'$, an initially disjoint copy of $G$, and the other one to a new vertex, $u_e'$
which is then in turn joined to $G_e'$. Finally replicate all vertices outside of $G$, yielding in particular an edge-free copy $G_e''$ of $G$. Figure~\ref{fig:StabilizingGadget} in Appendix~\ref{app:StabilizingGadget} displays the relevant parts of the construction for a simple example with a singleton $S=\{e\}$. 

Note that this construction allows us to reduce the problem of deciding whether in a given selection of edges all of them are stable to \Stability by stabilizing all other edges. 
Moreover, it yields the following \ANDomega function for \Stability. This is stated in the following corollary, whose quite straightforward proof is deferred to Appendix~\ref{app:and-for-stability} due to the space constraints. 

\begin{corollary}\label{cor:and-for-stability}
Mapping $k$ graphs $G_1,\ldots,G_k$ to $G_1+\cdots+G_k$ with all join edges stabilized using the construction from Lemma~\ref{lem:StabilizingGadget} is an \ANDomega function for \Stability.
\end{corollary}

\section{\boldmath The Complexity of 
\texorpdfstring{\betaStability}{Vertex-Cover Stability} and
\texorpdfstring{\betaVertexStability}{Vertex-Cover Vertex-Stability}}
We will now examine the complexity of stability with respect to the vertex-cover number $\beta$. 

First, note that \betaVertexStability is trivially in \P as it consists of the empty graphs. 
\begin{theorem}\label{thm:beta-vertex-stability}
Only the empty graphs are $\beta$-vertex-stable.
\end{theorem} 
The easy proof is deferred to Appendix~\ref{app:beta-vertex-stability}. 
Turning to the smaller change of deleting only an edge instead of a vertex, the situation changes radically. 
We will prove with Theorem~\ref{thm:beta-stability} that determining whether a graph is $\beta$-stable is \ThetaTwo-complete.
An important ingredient to the proof is the following analogue to Lemma~\ref{lem:doubling}, which shows how to $\beta$-stabilize an arbitrary edge of a given graph.
The proof is deferred to Appendix~\ref{app:BetaStabilizingGadgetProof} due to the space constraints.

\begin{lemma}\label{lem:BetaStabilizingGadget}
Let $G$ be a graph and $\{v_1,v_2\}\in E(G)$ one of its edges. Create from $G$ a new graph $G'$ by replacing the edge $\{v_1,v_2\}$ by the gadget that consists of four new vertices $u_1$, $u_2$, $u_3$, and $u_4$ with edges $\{u_1,u_2\}$, $\{u_2,u_3\}$, $\{u_3,u_4\}$, and $\{u_4,u_1\}$ (i.e., a new rectangle) and additionally the edges $\{v_1,u_1\}$, $\{v_1,u_3\}$, $\{v_2,u_2\}$, and $\{v_2,u_4\}$. 
(This gadget is displayed in Figure~\ref{fig:BetaStabilizingGadgetB} in Appendix~\ref{app:BetaStabilizingGadgetProof}.)
Then we have $\beta(G')	= \beta(G)+2$, all edges of the gadget are stable in $G'$, and the remaining edges are stable in $G'$ if and only if they are stable in $G$.
\end{lemma} 

\begin{theorem}\label{thm:beta-stability}
Determining whether a graph is $\beta$-stable is \ThetaTwo-complete.
\end{theorem}

\begin{proof}
We reduce from $\{(G,H)\in\mathcal{G}^2\mid \beta(G)>\beta(H)\}$, which is \ThetaTwo-hard~\cite[Thm.~12]{spa-vog:c:theta-two-classic}. (Note that this language is essentially the complement of $\alternativetextsc{CompareVertexCover}$ and that \ThetaTwo is closed under taking the complement.) 
Let $G$ and $H$ be given graphs. 
Replace each edge $e\in E(G)$ by a copy of the stabilizing gadget described in Lemma~\ref{lem:BetaStabilizingGadget}.  
Call the resulting graph $G'$. Clearly, we have $\card{V(G')} = \card{V(G)} + 4\card{E(G)}$. By Lemma~\ref{lem:BetaStabilizingGadget}, $G'$ is $\beta$-stable and $\beta(G') = \beta(G) + 2\card{E(G)}$. 
Moreover, let $H' = H \cup K_2$. The edge in $K_2$ ensures that $H'$ is not $\beta$-stable. Moreover, we have $\beta(H') = \beta(H)+1$ and $\card{V(H')}=\card{V(H)}+2$. 

Now, let
$G'' = G'$, just for consistent notation, and $H'' = H' \cup K_{2\card{E(G)}}$.  
Since $\beta(K_n)=n-1$ for $n\ge 1$, this implies $\beta(G'')-\beta(G) = \card{E(G)} = \beta(H'')-\beta(H)$. 
We finish the construction by adding isolated vertices to either $G''$ or $H''$ such that we achieve an equal number of vertices without changing the vertex cover number; that is, we let 
$G''' ={} G'' \cup I_{\max\{0,\card{V(H'')}-\card{V(G'')}\}}\text{ and }
H''' ={} H'' \cup I_{\max\{0,\card{V(G'')}-\card{V(H'')}\}}$.
Let $c = \card{V(G''')} = \card{V(H''')}$ and $d = \beta(G''')-\beta(G) = \beta(H''')-\beta(H)$. Note that $G'''$ is $\beta$-stable since we stabilized $G'$ with the gadget substitutions and then only added isolated vertices but no more edges. Moreover, $H'''$ is not $\beta$-stable due to the $\beta$-critical edge of $K_2$.

Let $S$ be the join $G'''+ H'''$ with all join edges stabilized, again by the gadget substitution described in
Lemma~\ref{lem:BetaStabilizingGadget}. 
It is easy to see from the proof of Lemma~\ref{lem:BetaStabilizingGadget} that the gadget as a whole behaves just like the edge it replaces, in the sense that an optimal vertex cover of the whole graph contains, without loss of generality, either $v_1$ or $v_2$ or both. Therefore, an optimal vertex cover of $S$ consists of either an optimal vertex cover of $G'''$ and all vertices of $H'''$ or of an optimal vertex cover of $H'''$ and all vertices of $G'''$ plus, in both cases, a constant number $k$ of vertices for covering the gadget edges -- namely two for each former join edge, that is, $k = 2 \cdot \card{V(G''')} \cdot \card{V(H''')}$. 
In the first case, we obtain an optimal vertex cover for $S$ of size $\beta(G''')=\beta(G)+d+c+k$, in the second case one of size $\beta(H''')=\beta(H)+d+c+k$. 

Assume first that $\beta(G)>\beta(H)$. It follows that $\beta(G''')>\beta(H''')$ and thus any optimal vertex cover for $S$ consists of all vertices $V(H''')$, an optimal vertex cover for $G'''$, and $k$ vertices for the gadgets. Since we ensured that $G'''$ is $\beta$-stable, $S$ is $\beta$-stable.
Now, assume that $\beta(G)\le \beta(H)$. Then there is an optimal vertex cover that consists of all vertices of $G'''$, an optimal vertex cover of $H'''$, and again $k$ vertices due to the gadgets. Since $H'''$ not  $\beta$-stable, as pointed out above, $S$ is not $\beta$-stable either. 
We conclude that $S$ is $\beta$-stable exactly if $\beta(G)>\beta(H)$, thus proving that $\beta$-stability is \ThetaTwo-hard and therefore \ThetaTwo-complete. 
\end{proof}

\section{Unfrozenness}
We begin with the observation that both for Colorability and for Vertex Cover adding a vertex is too generous a modification to be interesting. The trivial proof is found in Appendix~\ref{app:v-unfrozenness}.
\begin{theorem}\label{thm:v-unfrozenness}
There is no vertex-unfrozen graph and only one $\beta$-vertex-unfrozen graph, namely the null graph (i.e., the graph with the empty vertex set).
\end{theorem}

Both problems are far more interesting in the default setting, that is, for adding edges. 
The \ThetaTwo-completeness of deciding whether a given graph is $\beta$-unfrozen can be obtained by a method similar to the one we used to establish Theorem~\ref{thm:beta-stability}; see Appendix~\ref{app:beta-unfrozenness} for the proof. 
\begin{theorem}\label{thm:beta-unfrozenness}
Determining whether a graph is $\beta$-unfrozen is \ThetaTwo-complete.
\end{theorem}

Now, we would like to show the analogous result that \Unfrozenness is \ThetaTwo-complete as well. This turns out to be a very difficult task, however. 
There are many clues suggesting the hardness of \Unfrozenness, which exhibits a far richer structure than all of the problems listed in Table~\ref{tab:OverviewTable} as easy. The latter problems are either empty or singletons or consist of all independent sets or all cliques, while \Unfrozenness contains large classes of different graphs. We can even produce arbitrarily many new complicated unfrozen graphs using the graph join. There are no clearly identifiable characteristics to these unfrozen graphs to be leveraged. 
Instead, we give a sufficient condition for the \ThetaTwo-completeness of \Unfrozenness, namely the existence of a polynomial-time computable construction that turns arbitrary graphs into unfrozen ones without changing their chromatic number in an intractable way. 
\begin{theorem}\label{thm:unfrozenness}
Assume that there are polynomial-time computable functions  $f\colon \mathcal{G}\to \mathcal{G}$ and $g\colon \mathcal{G}\to \Z$ such that for any graph $G$ we have that $f(G)$ is unfrozen and $\chi(f(G))=\chi(G)+g(G)$. 
Then \Unfrozenness is \ThetaTwo-complete.
\end{theorem}
The proof of Theorem~\ref{thm:unfrozenness} is deferred to Appendix~\ref{app:unfrozenness}. It is similar in flavor to the one of Theorem~\ref{thm:beta-unfrozenness} and reduces from  $\alternativetextsc{CompareColorability}$, whose \ThetaTwo-hardness is stated now. 
\begin{theorem}\label{thm:CompareColorability}
$\alternativetextsc{CompareColorability}=\{(G,H)\in\mathcal{G}^2\mid \chi(G)\le\chi(H)\}$ is \ThetaTwo-hard. 
\end{theorem}
Theorem~\ref{thm:CompareColorability} is proved essentially in the same way as Wagner~\cite[Thm.~6.3.2]{wag:j:more-on-bh} proves the \ThetaTwo-hardness of \alternativetextsc{OddColorability}. 
As he suggests~\cite[page~79]{wag:j:more-on-bh}, it is 
rather straightforward to translate the hardness result 
for \alternativetextsc{OddColorability} into one for \alternativetextsc{EqualColorability}. 
This holds true for \alternativetextsc{CompareColorability} as well. 
The method for obtaining these results is easily generalized to yield 
two sufficient criteria for \ThetaTwo-hardness, stated as Lemmas~\ref{lem:equal} and~\ref{lem:compare} 
in Appendix~\ref{app:CompareColorability}. 
We use the latter lemma -- stated in a somewhat flawed form by 
Spakowski and Vogel~\cite[Lem.~9]{spa-vog:c:theta-two-classic} -- to 
prove Theorem~\ref{thm:CompareColorability}. 
See Appendix~\ref{app:CompareColorability} for all details.

Note that an analogue to Lemma~\ref{lem:StabilizingGadget} for unfreezing instead of stabilizing edges would be sufficient to satisfy the assumption of Theorem~\ref{thm:unfrozenness}. However, based on our efforts we suspect that a suitable gadget -- if one exists -- must be of significantly higher complexity than the one in Figure~\ref{fig:StabilizingGadget} in Appendix~\ref{app:StabilizingGadget}. 

\section{Two-Way Stability}
A graph is two-way stable if it is stable with respect to both the deletion and addition of an edge. 
First, we note that the analogous problem with respect to vertices is trivial for both Colorability and Vertex Cover. 
The following is an immediate consequence of Theorem~\ref{thm:v-unfrozenness}. 

\begin{theorem}\label{thm:v-twowaystability}
There is no vertex-two-way-stable graph and only one $\beta$-vertex-two-way-stable graph, namely the null graph with the empty vertex set.
\end{theorem}

The default case of edge deletion is more interesting. We begin with Colorability. 

\begin{theorem}\label{thm:twowaystability}
The problem \TwoWayStability is \ThetaTwo-complete.
\end{theorem}
To prove this, it suffices to check that mapping a graph $G$ to $G\cup G$ reduces \Unfrozenness to \TwoWayStability; see Appendix~\ref{app:twowaystability} for the details. 
We are able to prove the analogous result for $\beta$-\TwoWayStability via Lemma~\ref{lem:quadripartite-gadget}; the proof is deferred to Appendix~\ref{app:quadripartite-gadget}.
\enlargethispage{\baselineskip}
\begin{lemma}\label{lem:quadripartite-gadget}
	Let a nonempty graph $G$ and an edge $e\in V(G)$ be given. Construct from $G$ a graph $G'$ by substituting for $e$ the constant-size gadget that consists of a clique on the new vertex set $\{u_1,u_2,u_3,u_4,u_1',u_2',u_3',u_4'\}$, with the four edges $\{u_i,u_i'\}$ for $i\in\{1,2,3,4\}$ removed and the four edges $\{v,u_1\}$, $\{v,u_2\}$, $\{v',u_3\}$, and $\{v',u_4\}$ added. (This gadget is displayed in Figure~\ref{fig:TwoWayGadgetB} in Appendix~\ref{app:quadripartite-gadget}.) 
	The graph $G'$ has the following properties. 
\begin{enumerate}
	\item 
	$\beta(G')=\beta(G)+6$, 
	\item 
	every edge $e'\in E(G)\setminus\{e\}$ is $\beta$-stable in $G$ exactly if it is in $G'$, 
	\item 
	all remaining edges 
	of $G'$ 
	are $\beta$-stable,
	\item 
	every nonedge $e'\in\coE(G)$ is $\beta$-unfrozen in $G$ exactly if it is in $G'$, and
	\item 
	all remaining nonedges 
	$e'\in\coE(G')\setminus\coE(G)$ 
	of $G'$ 
	are $\beta$-unfrozen. 
\end{enumerate} 
\end{lemma}

An iterated application of this lemma allows us to stabilize an arbitrary set of edges of an arbitrary graph without introducing any new unfrozen edges. 
The \ThetaTwo-hardness of $\beta$-\TwoWayStability is now an easy consequence of Lemma~\ref{lem:quadripartite-gadget};  see Appendix~\ref{app:beta-twowaystability}. 
\begin{theorem}\label{thm:beta-twowaystability}
	The problem $\beta$-\TwoWayStability is \ThetaTwo-complete.
\end{theorem}

\bibliography{mybib}
\newpage
\begin{appendix}
	
\section{Proofs of All Observations}\label{app:observations}
We restate and prove our four  Observations~\ref{obs:differsbyatmostone}, \ref{obs:CriticalityFromEdges}, \ref{obs:StabilityFromVertices}, and \ref{obs:CriticalVertexCharacterization}.

\setcounter{tempcounter}{\value{theorem}}
\setcounterref{theorem}{obs:differsbyatmostone}
\addtocounter{theorem}{-1}
\begin{observation}
	The deletion of an edge or of a vertex either decreases the chromatic number by exactly one or leaves it unchanged. 
\end{observation}
\setcounter{theorem}{\value{tempcounter}}
\begin{proof}
	It is clear that deleting an edge or vertex cannot increase the chromatic number. 
	To see that this cannot decrease it by more than one, it suffices to 
	note that introducing a new edge or vertex can increase it by at
	most one since we can assign one new, unique color to the new vertex
	or to one of the two vertices of the inserted edge.
\end{proof}
\setcounter{tempcounter}{\value{theorem}}
\setcounterref{theorem}{obs:CriticalityFromEdges}
\addtocounter{theorem}{-1}
\begin{observation}
	Let $e=\{u,v\}$ be a critical edge. Then $u$ and $v$ are critical as well. 
\end{observation}
\setcounter{theorem}{\value{tempcounter}}
\begin{proof} 
	Since $G-u$ and $G-v$ are subgraphs of $G-e$, both $\chi(G-u)$ and $\chi(G-v)$ are at most $\chi(G-e)$, which is less than $\chi(G)$ because $e$ is critical. Thus $u$ and $v$ are critical.
\end{proof}
\setcounter{tempcounter}{\value{theorem}}
\setcounterref{theorem}{obs:StabilityFromVertices}
\addtocounter{theorem}{-1}
\begin{observation}
	Let $v$ be a stable vertex. Then all edges incident to $v$ are stable.  
\end{observation}
\setcounter{theorem}{\value{tempcounter}}
\begin{proof}
	This follows immediately from the contrapositive of Observation~\ref{obs:CriticalityFromEdges}.
\end{proof}
\setcounter{tempcounter}{\value{theorem}}
\setcounterref{theorem}{obs:CriticalVertexCharacterization}
\addtocounter{theorem}{-1}
\begin{observation}
	Let $G$ be a graph. A vertex $v\in V(G)$ is critical if and only if there is an optimal coloring of $G$ that assigns $v$ a color with which no other vertex is colored. 
\end{observation}
\setcounter{theorem}{\value{tempcounter}}
\begin{proof}
	Given a critical $v\in V(G)$, consider an arbitrary optimal coloring of $G-v$. Since $v$ is critical, it uses one fewer color than the optimal colorings of $G$. We therefore obtain an optimal coloring of $G$ by assigning $v$ a new color. 
	The converse is immediate. 
\end{proof}

\section{Proof of Proposition~\ref{pro:CliqueAndIS}}
\label{app:CliqueAndIS}
We detail how to derive results for Clique and Independent Set summarized in Table~\ref{tab:OverviewTable} from the results for Vertex Cover and vice versa. 
We restate and prove Proposition~\ref{pro:CliqueAndIS}.
\setcounter{tempcounter}{\value{theorem}}
\setcounterref{theorem}{pro:CliqueAndIS}
\addtocounter{theorem}{-1}
\begin{proposition}
Let $\coG$ denote the complement graph of $G$. We have the following equalities.
\begin{enumerate}
	\item 
	$\beta\text{-}\Stability=\alpha\text{-}\Stability=\{\coG\mid G\in \omega\text{-}\Unfrozenness\}$.
	\item 
	$\beta\text{-}\Unfrozenness=\alpha\text{-}\Unfrozenness=\{\coG\mid G\in \omega\text{-}\Stability\}$.
	\item 
	$\beta\text{-}\TwoWayStability=\alpha\text{-}\TwoWayStability=\{\coG\mid G\in \omega\text{-}\TwoWayStability\}$.
	\item $\beta\text{-}\VertexStability=\{I_n\mid n\in\N\}$.
	\item $\alpha\text{-}\VertexStability=\{\coG \mid G\in\omega\text{-}\VertexStability\}$. 
	\item
	$\beta\text{-}\VertexUnfrozenness=\beta\text{-}\VertexTwoWayStability=\{K_0\}$.
	\item $\alpha\text{-}\VertexUnfrozenness=\alpha\text{-}\VertexTwoWayStability=\\
	\omega\text{-}\VertexUnfrozenness=\omega\text{-}\VertexTwoWayStability=\emptyset$. 
\end{enumerate} 
\end{proposition}
\setcounter{theorem}{\value{tempcounter}}
\begin{proof}
	For the second equality of the first three items it suffices to note that an independent set of a graph is a clique of its complement graph and vice versa. 
	The first equality of the first three items follows from the fact that, on the one hand, for any graph on $n$ vertices, the complement of a vertex cover of size $k$ is an independent set of size $n-k$ and, on the other hand, adding or deleting edges obviously does not change the number of vertices. 
	For the remaining items, we add or delete vertices, so this argument does not hold anymore. Item 4 is exactly  Theorem~\ref{thm:beta-vertex-stability}.  
    For item 5, we simply use that a clique is an independent set in the complement graph and vice versa. Item 6 combines Theorems~\ref{thm:v-unfrozenness} and~\ref{thm:v-twowaystability}. Item 7 finally follow from the fact that adding an isolated vertex increases $\alpha$, while adding a universal vertex increases $\omega$.
\end{proof}

\section{Proof of Proposition~\ref{pro:StabilityAndCriticality}}
\label{app:StabilityAndCriticality}
We restate and prove Proposition~\ref{pro:StabilityAndCriticality}.
\setcounter{tempcounter}{\value{theorem}}
\setcounterref{theorem}{pro:StabilityAndCriticality}
\addtocounter{theorem}{-1}
\begin{proposition}
We have the following equalities. 
\begin{enumerate}
	\item $\betaVertexStability=\alpha\text{-}\VertexCriticality=\{\coG\mid G\in \omega\text{-}\VertexCriticality\}$. 
	\item $\beta\text{-}\VertexCriticality=\alpha\text{-}\VertexStability=\{\coG\mid G\in \omega\text{-}\VertexStability\}$. 
\end{enumerate}
\end{proposition}
\setcounter{theorem}{\value{tempcounter}}
\begin{proof}
It suffices to prove that a vertex $v$ of a graph $G$ is $\alpha$-stable if and only if it is $\beta$-critical. 
In the forth step of the following equivalence chain we use that every minimum vertex cover is the complement of a maximum independent set and vice versa. 

\begin{align*}
v\text{ is $\alpha$-stable}\ 
&\Longleftrightarrow\ \alpha(G-v)=\alpha(G)\\
&\Longleftrightarrow\ \card{V(G)}-\alpha(G-v)=\card{V(G)}-\alpha(G)\\
&\Longleftrightarrow\ \card{V(G-v)}-\alpha(G-v)+1=\card{V(G)}-\alpha(G)\\
&\Longleftrightarrow\ \beta(G-v)+1=\beta(G)\\
&\Longleftrightarrow\ \beta(G-v)\neq\beta(G)\\
&\Longleftrightarrow\ v\text{ is $\beta$-critical}. 
\end{align*}
\end{proof}

\section{Proof of Lemma~\ref{lem:v-stability-to-stability}}
\label{app:v-stability-to-stability}
We restate and prove Lemma~\ref{lem:v-stability-to-stability}.
\setcounter{tempcounter}{\value{theorem}}
\setcounterref{theorem}{lem:v-stability-to-stability}
\addtocounter{theorem}{-1}
\begin{lemma}
	\VertexStability polynomial-time many-one reduces to \Stability.
\end{lemma}
\setcounter{theorem}{\value{tempcounter}}
\begin{proof}
	We show that $G$ is vertex-stable if and only if $G + G$ is stable.
	It is immediate that if $G$ is vertex-stable, then
	$G + G$ is vertex-stable, and thus $G + G$ is stable by Observation~\ref{obs:StabilityFromVertices}. 
	For the converse, suppose that $G$ is not vertex-stable.
	Then there is a vertex $v$ with $\chi(G - v) = \chi(G) - 1$. 
	Fix any optimal coloring of $G-v$. 
	We now color both copies of $G-v$ in $G + G$ according to it, but using two disjoint sets of $\chi(G)-1$ colors. 
	Finally, we assign one additional new color 
	to the two vertices corresponding to $v$. 
	This colors the graph $G + G$ from which the edge between the two copies of $v$ has been deleted
	with $2(\chi(G) - 1) + 1 = \chi(G + G) - 1$ colors, proving that 
	$G + G$ is not stable.
\end{proof}

\section{Proof of Theorem~\ref{thm:JoinIsANDomega}}
\label{app:JoinIsANDomega}
We restate and prove Theorem~\ref{thm:JoinIsANDomega}.
\setcounter{tempcounter}{\value{theorem}}
\setcounterref{theorem}{thm:JoinIsANDomega}
\addtocounter{theorem}{-1}
\begin{theorem}
	The join is an \ANDomega function for \VertexStability and \Unfrozenness.
\end{theorem}
\setcounter{theorem}{\value{tempcounter}}
\begin{proof}
	Let 
	$G_1,\ldots,G_n$ 
	be a finite number of graphs.  
	Consider $G_1+\cdots+G_n$. 
	For every $i\in\{1,\ldots,n\}$, the join edges force the vertices $V(G_i)$ to have colors that are different from the colors of all remaining vertices. This implies $\chi(G_1+\cdots+G_n)=\chi(G_1)+\cdots+\chi(G_n)$.
	Moreover, vertex deletion and edge addition commute with joining: 
	For every $v\in V(G_1+\cdots+G_n)=V(G_1)\cup\cdots\cup V(G_n)$, 
	there is an $i$ such that 
	$(G_1+\cdots+G_n)-v=G_1+\cdots+G_{i-1}+(G_i-v)+G_{i+1}+\cdots+G_n$. Analogously, for every nonedge  $e\in\coE(G_1+\cdots+G_n)=\coE(G_1)\cup \cdots\cup \coE(G_n)$, there is an $i$ such that  $(G_1+\cdots+G_n)+e=G_1+\cdots+G_{i-1}+(G_i+e)+G_{i+1}+\cdots+G_n$. 
	The claim of the theorem follows immediately. 
\end{proof}

\section{Proof of Lemma~\ref{lem:SatToThreeSat}}
\label{app:SatToThreeSat}
We restate and prove Lemma~\ref{lem:SatToThreeSat}.
\setcounter{tempcounter}{\value{theorem}}
\setcounterref{theorem}{lem:SatToThreeSat}
\addtocounter{theorem}{-1}
\begin{lemma}
	There is a polynomial-time many-one reduction from \Sat to \ThreeSat converting a \CNF-formula $\Phi$ into a \ThreeCNF-formula $\Psi$ such that 
	$\Phi$ is stable if and only if $\Psi$ is stable. 
\end{lemma}
\setcounter{theorem}{\value{tempcounter}}
\begin{proof}
	The standard clause-size reducing reduction maps a \CNF-formula $\Phi=C_1\wedge\cdots\wedge C_m$ to $\Psi'=F_1\wedge\cdots\wedge F_m$ by splitting any clause $C=(\ell_1\vee\cdots\vee\ell_k)$ with $k$ literals for a $k\ge 3$ into the $k$ clauses of the subformula $F=(\ell_1\vee y_1)\wedge(\overline{y}_1\vee \ell_2\vee y_2)\wedge\cdots\wedge(\overline{y}_{k-2}\vee \ell_{k-1}\vee y_{k-1})\wedge(\overline{y}_{k-1}\vee \ell_k)$ with fresh variables $y_i$ that do not occur elsewhere. 
	
	We prove that $\Phi$ is stable if and only if $\Psi'$ is. 
	We already know that $\Phi$ is satisfiable exactly if $\Psi'$ is. 
	Thus it suffices to show that $\Phi$ is satisfiable after deleting $C$ if and only if $\Psi'$ is satisfiable after deleting some clause in $F$. 
	Let $\alpha$ be a satisfying assignment for $\Phi-C$. 
	We then obtain a satisfying assignment $\beta$ for $\Psi'-(\ell_1\vee y_1)$ by setting $\beta(y_1)=\cdots=\beta(y_{k-1})=0$ and the remaining variables as usual. 
	For the converse, we simply observe that the restriction of an assignment satisfying $\Psi'$ with an arbitrary clause from $F$ deleted satisfies $\Phi-C$. 
	
	We now transform $\Psi'$ into a formula $\Psi\in\ThreeCNF$ that is stable exactly if $\Psi'$ is. We do this by substituting for any two-literal clause $(\ell_1\vee \ell_2)$ the subformula $(\ell_1\vee \ell_2\vee z)\wedge (\ell_1\vee \ell_2\vee \overline{z})$ and for any one-literal clause $(\ell_1)$ the subformula $(\ell_1\vee z_1\vee z_2)\wedge (\ell_1\vee z_1\vee \overline{z}_2)\wedge (\ell_1\vee \overline{z}_1\vee z_2)\wedge (\ell_1\vee \overline{z}_1\vee \overline{z}_2)$, where $z$, $z_1$, and $z_2$ are, for each substitution, new variables that do not occur anywhere else. It is now straightforward to check that $\Psi$ has all the desired properties. 
\end{proof}

\section{Proof That \texorpdfstring{\StableThreeCNF}{Stable3CNF} Has an OR\texorpdfstring{\textsubscript{2}}{2} Function}\label{app:or}
\begin{proof}
Let $\Phi  = C_1 \wedge \cdots \wedge C_m$ and  
$\Phi' = C'_1 \wedge \cdots \wedge C'_{m'}$ be two given \ThreeCNF-formulas. 
Without loss of generality, $\Phi$ and $\Phi'$ have disjoint variable sets. 
Let $\Psi = \bigwedge_{1\le i\le m,1\le j\le m'} (C_i \vee C'_j)$. Note that $\Psi$ is in \SixCNF and equivalent to $\Phi \vee \Phi'$. 
We will show that $\Psi \in  \StableSixCNF$ if and only if $\Phi \in \StableThreeCNF$ or
$\Phi' \in \StableThreeCNF$.
Setting the clause length of $\Psi$ to exactly $3$ by applying Lemma~\ref{lem:SatToThreeSat} then yields the desired \ORtwo-reduction.

First assume that neither $\Phi$ nor $\Phi'$ is in $\StableThreeCNF$. Then we have $\Phi, \Phi' \notin \ThreeSat$
and there are $\hat{\imath}\in\{1,\dots,m\}$ and $\hat{\jmath}\in\{1,\dots,m'\}$ and assignments $\assignment$ and
$\assignment'$ such that $\assignment$ satisfies $\Phi - C_{\hat{\imath}}$ and
$\assignment'$ satisfies
$\Phi' - C'_{\hat{\jmath}}$. Then we have $\Psi \notin \Sat$ and $(\assignment,\assignment')$
satisfies $(C_i \vee C'_j)$ for all
$(i,j) \neq (\hat{\imath},\hat{\jmath})$.
It follows that 
$\Psi - (C_{\hat{\imath}} \vee C'_{\hat{\jmath}})$ is satisfiable,
and thus $\Psi \notin \StableSixCNF$.

Now assume that $\Psi \notin \StableSixCNF$. Then $\Psi \notin \Sat$, and hence 
$\Phi, \Phi' \notin \ThreeSat$. 
There are indices $\hat{\imath}$ and $\hat{\jmath}$ such that
$\Psi - (C_{\hat{\imath}} \vee C'_{\hat{\jmath}})$ is satisfiable,
say by assignment $\assignmenttwo$. This $\assignmenttwo$ satisfies 
$(C_i \vee C'_j)$ for all $(i,j) \neq (\hat{\imath},\hat{\jmath})$.
In particular, $\assignmenttwo$ satisfies $(C_i \vee \Phi')$ for all 
$i \neq \hat{\imath}$ and $(\Phi \vee C'_j)$ for all $j \neq \hat{\jmath}$.
Since $\Phi, \Phi' \notin \ThreeSat$, this implies that $\assignmenttwo$ satisfies
$C_i$ for all $i \neq \hat{\imath}$ and $C'_j$ for all
$j \neq \hat{\jmath}$. It follows that $\Phi, \Phi' \notin \StableThreeCNF$.
\end{proof}

\section{Proof of Lemma~\ref{lem:cai-meyer}}
\label{app:cai-meyer}
We restate and prove Lemma~\ref{lem:cai-meyer}.

\setcounter{tempcounter}{\value{theorem}}
\setcounterref{theorem}{lem:cai-meyer}
\addtocounter{theorem}{-1}
\begin{lemma}
	A \ThreeCNF-formula $\Phi$ is not stable  
	if and only if
	$\chi(G_\Phi) > \chi(G_\Phi - t_{i1})$ for at least one $i\in\{1,\ldots,m\}$.
\end{lemma}
\setcounter{theorem}{\value{tempcounter}}
\begin{proof}[Proof of Lemma~\ref{lem:cai-meyer}]
	As stated by Cai and Meyer~\cite[Lem.~2.2]{cai-mey:j:dp}, $\Phi$ is satisfiable if and only if $G_\Phi$ is 3-colorable. 
	Note that a $\Phi$ that is not stable is not satisfiable. 
	If suffices to note $\Phi - C_i$ is satisfiable if and only if $G_\Phi - t_{i1}$ is 3-colorable. The mentioned paper proves the implication by picture~\cite[Figure~2.12]{cai-mey:j:dp} and states the converse in the second-to-last paragraph of is second section. 
\end{proof}

\section{Illustration of the Reduction From \EThreeSat to \ThreeCol}
\label{app:caimeyertotal}
In Lemma~\ref{lem:cai-meyer} and Theorem~\ref{thm:cai-meyer}, we rely on an extended version of the standard reduction from \ThreeSat to \ThreeCol. The construction is due to Cai and Meyer~\cite{cai-mey:j:dp}, who describe it in separate construction steps. 
For convenience, we provide in Figure~\ref{fig:caimeyertotal}~\cite[Figure 1, slightly modified]{bur-fre-etal:c:neighborly-help} an overview of the resulting, complete construction for a generic example.

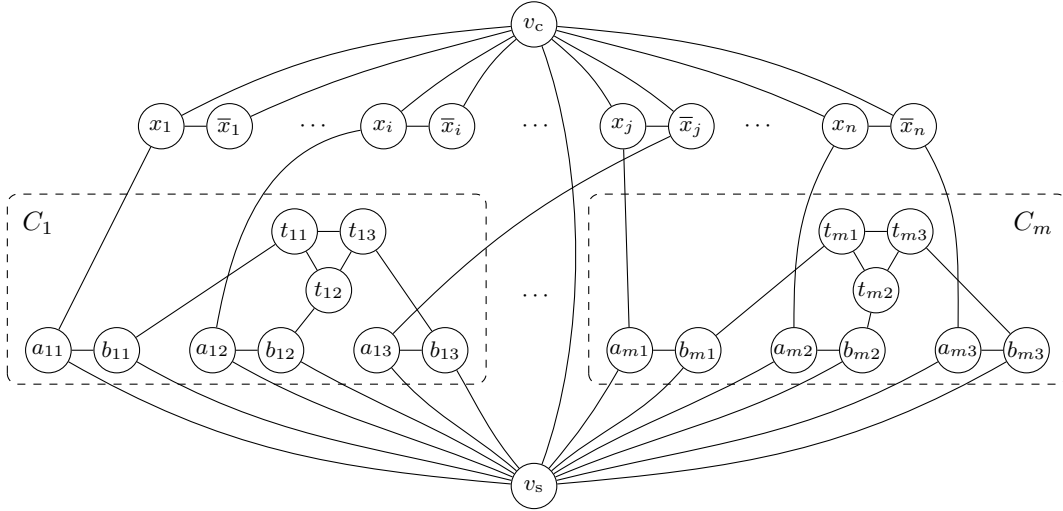
\begin{figure}[h]
	\begin{tikzpicture}[x=0.9cm,y=0.9cm,vertex/.style={draw,circle, inner sep=.0pt, minimum size=0.6cm},font=\footnotesize]
	\node[vertex] (vc) at (6.5,2.8) {$v_\textnormal{c}$};
	\node[vertex] (vs) at (6.5,-4) {$v_\textnormal{s}$};
	
	\node[vertex] (x_1) at (1.2*3-1+1.2*1+2.75-5.5,1.3) {$x_1$};
	\draw (x_1) edge [in=185, out=28] (vc);
	\node[vertex] (nx_1) at (1.2*3-1+1.2*1+2.75-4.5,1.3) {$\overline{x}_1$};
	\draw (nx_1) edge [in=195, out=25] (vc);
	\draw (nx_1) -- (x_1);
	\node at (1.2*3-1+1.2*1+2.75-3.25,1.3) {$\cdots$};
	
	\node[vertex] (x_i) at (1.2*3-1+1.2*1+2.75-2.25,1.3) {$x_i$};
	\draw (x_i) edge [in=210, out=40] (vc);
	\node[vertex] (nx_i) at (1.2*3-1+1.2*1+2.75-1.25,1.3) {$\overline{x}_i$};
	\draw (nx_i) edge [in=220, out=60] (vc);
	\draw (nx_i) -- (x_i);
	\node at (1.2*3-1+1.2*1+2.75,1.3) {$\cdots$};
	
	\node[vertex] (x_j) at (1.2*3-1+1.2*1+2.75+1.25,1.3) {$x_j$};
	\draw (x_j) edge [in=-40, out=120] (vc);
	\node[vertex] (nx_j) at (1.2*3-1+1.2*1+2.75+2.25,1.3) {$\overline{x}_j$};
	\draw (nx_j) edge [in=-30, out=140] (vc);
	\draw (nx_j) -- (x_j);
	\node at (1.2*3-1+1.2*1+2.75+3.25,1.3) {$\cdots$};
	
	\node[vertex] (x_n) at (1.2*3-1+1.2*1+2.75+4.5,1.3) {$x_n$};
	\draw (x_n) edge [in=-15, out=155] (vc);
	\node[vertex] (nx_n) at (1.2*3-1+1.2*1+2.75+5.5,1.3) {$\overline{x}_n$};
	\draw (nx_n) edge [in=-5, out=152] (vc);
	\draw (nx_n) -- (x_n);            
	
	\node[vertex] (a_11) at (2.4*1-3,-2) {$a_{11}$};
	\draw (a_11) edge[in=175,out=-30] (vs);
	\node[vertex] (b_11) at (2.4*1-2,-2) {$b_{11}$};
	\draw (b_11) edge[in=166,out=-30]  (vs);
	\draw (b_11) -- (a_11);
	
	\node[vertex] (a_21) at (2.4*2-3,-2) {$a_{12}$};
	\draw (a_21) edge[in=157,out=-30] (vs);
	\node[vertex] (b_21) at (2.4*2-2,-2) {$b_{12}$};
	\draw (b_21) edge[in=148,out=-30]  (vs);
	\draw (b_21) -- (a_21);
	
	\node[vertex] (a_31) at (2.4*3-3,-2) {$a_{13}$};
	\draw (a_31) edge[in=139,out=-50] (vs);
	\node[vertex] (b_31) at (2.4*3-2,-2) {$b_{13}$};
	\draw (b_31) edge[in=130,out=-60]  (vs);
	\draw (b_31) -- (a_31);
	
	\node[vertex] (a_12) at (2.4*1 +5.5,-2) {$a_{m1}$};
	\draw (a_12) edge[in=50,out=180+60] (vs);
	\node[vertex] (b_12) at (2.4*1 +6.5,-2) {$b_{m1}$};
	\draw (b_12) edge[in=41,out=180+50] (vs);
	\draw (b_12) -- (a_12);
	
	\node[vertex] (a_22) at (2.4*2 +5.5,-2) {$a_{m2}$};
	\draw (a_22) edge[in=32,out=180+30] (vs);
	\node[vertex] (b_22) at (2.4*2 +6.5,-2) {$b_{m2}$};
	\draw (b_22) edge[in=23,out=180+30] (vs);
	\draw (b_22) -- (a_22);
	
	\node[vertex] (a_32) at (2.4*3 +5.5,-2) {$a_{m3}$};
	\draw (a_32) edge[in=14,out=180+30] (vs);
	\node[vertex] (b_32) at (2.4*3 +6.5,-2) {$b_{m3}$};
	\draw (b_32) edge[in=5,out=180+30] (vs);
	\draw (b_32) -- (a_32);
	
	\node[vertex] (t11) at (3,-0.25) {$t_{11}$};
	\draw (t11)--(b_11);
	\node[vertex] (t12) at (3.5,-1.116) {$t_{12}$};
	\draw (t12)--(b_21);
	\draw (t12)--(t11);
	\node[vertex] (t13) at (4,-0.25) {$t_{13}$};
	\draw (t13)--(b_31);
	\draw (t12)--(t13);
	\draw (t13)--(t11);
	\node[vertex] (t21) at (11,-0.25) {$t_{m1}$};
	\draw (t21)--(b_12);
	\node[vertex] (t22) at (11.5,-1.116) {$t_{m2}$};
	\draw (t22)--(b_22);
	\draw (t22)--(t21);
	\node[vertex] (t23) at (12,-0.25) {$t_{m3}$};
	\draw (t23)--(b_32);
	\draw (t22)--(t23);
	\draw (t23)--(t21);
	
	\node at (1.2*3-1+1.2*1+2.75,-1.2) {$\cdots$};
	
	\draw (x_1)--(a_11);
	\draw (x_i) edge [in=80, out=190](a_21);
	\draw (x_j)--(a_12);
	\draw (x_n)edge [in=90, out=-120](a_22);
	\draw (nx_j)edge[bend right=10](a_31);
	\draw (nx_n)edge [in=90, out=-60](a_32);
	
	\draw (vc)edge [in=70, out=-70, looseness=.8](vs);
	
	\draw[dashed, rounded corners] (-1.2, -2.5) rectangle (5.8, .3) {};
	\draw[dashed, rounded corners] (7.3, -2.5) rectangle (14.3, .3) {};
	\node[font=\normalsize] at (-.75,-.12) {$C_1$};
	\node[font=\normalsize] at (13.8,-.12) {$C_m$};
	\end{tikzpicture}
	
	\caption{The graph $G_\Phi$ for a \EThreeCNF-formula with 
		$C_1 = x_1 \vee x_i \vee \overline{x}_j$ and 
		$C_m = x_j \vee x_n \vee \overline{x}_n$.}\label{fig:caimeyertotal}
\end{figure}

\section{Proof of Lemma~\ref{lem:StabilizingGadget}}
\label{app:StabilizingGadget}

We restate and prove Lemma~\ref{lem:StabilizingGadget}. 
\setcounter{tempcounter}{\value{theorem}}
\setcounterref{theorem}{lem:StabilizingGadget}
\addtocounter{theorem}{-1}
\begin{lemma}
	There is a polynomial-time algorithm that, given any graph $G$ plus a nonempty subset $S\subseteq E(G)$ of its edges, adds a fixed gadget to the graph and then substitutes for every $e\in S$ some gadget that depends on $G$ and $e$, yielding a graph $\widehat{G}$ with the following properties:
	\begin{enumerate}
		\item 
		$\chi(\widehat{G})=\chi(G)+2$.
		\item 
		All edges in $E(\widehat{G})\setminus (E(G)\setminus S)$ are stable. 
		\item Each one of the remaining edges in $E(G)\setminus S$ is stable in $\widehat{G}$ exactly if it is stable in $G$.
	\end{enumerate}
\end{lemma} 
\setcounter{theorem}{\value{tempcounter}}
	
\begin{figure}[htbp]
	\begin{center}
		\begin{subfigure}[c]{0.3\textwidth}
			\centering
			\begin{tikzpicture}[
			vertex/.style={draw,circle, inner sep=.0pt, minimum size=0.53cm},font=\scriptsize]
			\draw[white] (0,0.7) ellipse (1cm and 2.5cm); 
			\draw (.5,0) ellipse (1.8cm and .6cm);
			
			\begin{scope}[shift={(.1,0)}]	
			\node[vertex] (v1) at (-.6,0) {$v_1$};
			\node[vertex] (v2) at (.6,0) {$v_2$};
			
			\draw (v1) -- (v2);
			\begin{scope}
			\draw[dashed] (0,0) ellipse (1.2cm and .4cm);
			\clip (0,0) ellipse (1.2cm and .4cm);
			
			\draw[] (v1) -- +(170:1);
			
			\draw[] (v2) -- +(-40:1);
			\draw[] (v2) -- +(30:1);
			\end{scope}
			\end{scope}
			\node (G) at (1.8,0) {$G$};
			\end{tikzpicture}
			\captionsetup{width=.9\linewidth}
			\subcaption{An example of a simple given graph $G$ on five vertices with only a single edge $e=\{v_1,v_2\}$ to be stabilized. The vertices outside the dashed ellipse and edge parts leading to them are omitted.}
			\label{fig:StabilizingGadgetA}
		\end{subfigure}
		\begin{subfigure}[c]{0.68\textwidth}
			\centering
			\begin{tikzpicture}[
			xscale=1.05,yscale=.99,
			vertex/.style={draw,circle, inner sep=.0pt, minimum size=0.53cm},font=\scriptsize]
			
			\coordinate (h1) at (-3.3,0.2);
			\coordinate (h4) at (4,1);
			\coordinate (h5) at (3.5,-.5);
			\node[vertex] (w2') at (-.3,-2.5) {$w_2'$};
			\node[vertex] (w1') at (-.1,-1.9) {$w_1'$};
			\node[vertex] (w1'') at (2.1,-1.9) {$w_1''$};
			\node[vertex] (w2'') at (2.3,-2.5) {$w_2''$};
			\draw (w1') -- (h1) -- (w2');
			\draw (w1') -- (h4) -- (w2');
			\draw (w1') -- (h5) -- (w2');
			
			\draw (w1'') -- (h1) -- (w2'');
			\draw (w1'') -- (h4) -- (w2'');
			\draw (w1'') -- (h5) -- (w2'');
			
			\draw[fill=white] (1,0) ellipse (4.5cm and 1.4cm);
			
			\node[vertex] (v1) at (-2.5,0) {$v_1$};
			\node[vertex] (v2) at (3.5,0) {$v_2$};
			
			\begin{scope}
			\draw[dashed] (.5,0) ellipse (3.5cm and 1cm);
			\clip (.5,0) ellipse (3.5cm and 1cm);
			
			\draw[] (v1) -- +(170:1);
			
			\draw[] (v2) -- +(-40:1);
			\draw[] (v2) -- +(30:1);
			\end{scope}
			
			\coordinate (h'1) at (-1.4,.4);
			\coordinate (h'3) at (1.3,-.5);
			\coordinate (h'4) at (1.7,-.8);
			
			\begin{scope}[shift={(-.7,2.4)},rotate=30]
			\coordinate (h''1) at (-1.4,.4);
			\coordinate (h''3) at (1.3,-.5);
			\coordinate (h''4) at (1.7,-.6);
			\end{scope}
			
			\node[vertex] (u') at (2.0,.6) {$u_e'$};
			\node[vertex] (u'') at (1.7,2) {$u_e''$};
			\draw[] (u') -- (h'1) -- (v1);
			\draw[] (u') -- (h'3) -- (v1);
			\draw[] (u') -- (h'4) -- (v1);
			
			\draw[] (u'') -- (h'1);
			\draw[] (u'') -- (h'3);
			\draw[] (u'') -- (h'4);
			
			\begin{scope}[shift={(.2,-.1)},rotate=-18]
			\draw[fill=white] (0,0) ellipse (1.8cm and .6cm);
			
			\begin{scope}[shift={(-.4,0)}]		
			\node[vertex] (v1') at (-.6,0) {$v_{1,e}'$};
			\node[vertex] (v2') at (.6,0) {$v_{2,e}'$};
			\draw (v1') -- (v2');
			\draw (v1') -- (u'');
			\draw (v2') -- (u'');
			\begin{scope}
			\draw[dashed] (0,0) ellipse (1.1cm and .4cm);
			\clip (0,0) ellipse (1.1cm and .4cm);
			
			\draw[] (v1') -- +(170:1);
			
			\draw[] (v2') -- +(-40:1);
			\draw[] (v2') -- +(30:1);
			\end{scope}	
			\end{scope}
			\end{scope}	
			
			\begin{scope}[shift={(0,2.5)}]
			\draw[] (u') -- (h''1) -- (v1);
			\draw[] (u') -- (h''3) -- (v1);
			\draw[] (u') -- (h''4) -- (v1);	
			
			\begin{scope}[shift={(0.2,-.2)}]
			\begin{scope}[shift={(-1.2,-.1)},rotate=18]
			\draw[fill=white] (.5,0) ellipse (1.8cm and .6cm);
			
			\begin{scope}[shift={(.1,0)}]	
			\node[vertex] (v1'') at (-.6,0) {$v_{1,e}''$};
			\node[vertex] (v2'') at (.6,0) {$v_{2,e}''$};
			
			\begin{scope}
			\draw[dashed] (0,0) ellipse (1.1cm and .4cm);
			\clip (0,0) ellipse (1.1cm and .4cm);		
			
			\end{scope}
			\end{scope}
			\end{scope}
			\end{scope}
			\end{scope}	
			
			\draw (u') -- (v2);
			\draw (u'') -- (v2);
			\draw (u') -- (v1') -- (v1);
			\draw (u') -- (v2') -- (v1);
			\draw (u') -- (v1'') -- (v1);
			\draw (u') -- (v2'') -- (v1);
			\draw (v1') -- (v2'');
			\draw (v2') -- (v1'');
			
			\begin{scope}
			\clip (0,2.5) ellipse (2cm and .6cm);
			\draw (u') -- (v2'');
			\end{scope}
			
			\node (G) at (5,0) {$G$};
			\node (G') at (1.4,-.5) {$G_e'$};
			\node (G'') at (0.7,2.7) {$G_e''$};	
			
			\draw (w1') -- (w2');
			\draw (w1'') -- (w2'');
			\draw (w1') -- (w2'');
			\draw (w1'') -- (w2');
			
			\draw (v1) -- (w1') -- (v2);
			\draw (v1) -- (w2') -- (v2);
			\draw (v1) -- (w1'') -- (v2);
			\draw (v1) -- (w2'') -- (v2);
			\end{tikzpicture}
			\captionsetup{width=1\linewidth}
			\subcaption{The graph $\widehat{G}$, satisfying all properties stated in Lemma~\ref{fig:StabilizingGadget}. In the simple case of our example that $S$ is only a singleton $\{e\}$, the construction contains three partly modified copies of $G$, with the same vertices and edge parts as in Figure~\ref{fig:StabilizingGadgetA} omitted from the picture. 
			}
		\end{subfigure}
		\caption{Example of the construction of Lemma~\ref{lem:StabilizingGadget}, rendering arbitrary subsets of edges stable.
		}\label{fig:StabilizingGadget}
	\end{center}
\end{figure}
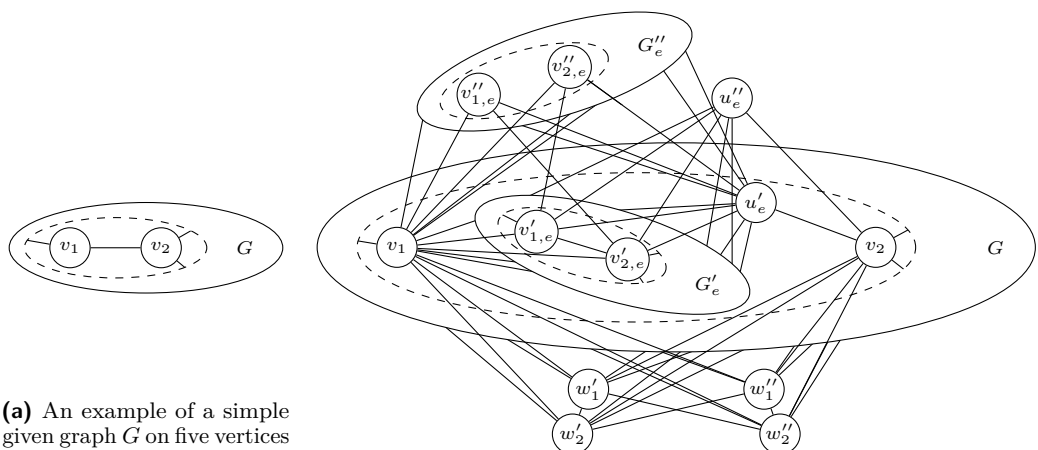    

\begin{proof}
	Let a graph $G$ and a subset $S\subseteq E(G)$ of edges in it be given. 
	We first describe the construction of $\widehat{G}$ in detail. Figure~\ref{fig:StabilizingGadget} exemplifies the full construction for a simple graph $G$ on five vertices and an $S$ that contains exactly one edge $e=\{v_1,v_2\}$. 
	
	We begin by adding a cycle on four new vertices $w_1'$, $w_2'$, $w_1''$, $w_2''$ and joining it to $G$. 
	Then we complete the procedure described in the following paragraph for each edge $e\in S$. 
	
	We add in disjoint union a copy $G'_e$ of the original graph $G$. We distinguish the vertices of $G'_e$ from the ones of $G$ by adding a prime and the subindex $e$ to them. In the example, where the edge $e=\{v_1,v_2\}$ is to be stabilized, there will thus be two adjacent vertices $v_{1,e}',v_{2,e}'\in V(G'_e)$, as shown in Figure~\ref{fig:StabilizingGadget}. 
	Now we join two new vertices to $G_e'$. One of them we merge with $v_1$; the other one we call $u_e'$ and connect it to $v_2$. 
	We then replicate $u_e'$ and every $v'\in V(G_e')$, marking the replicas with a second prime. 
	Excluding $u''$, these replicas constitute another, empty copy of $G$, which we denote $G_e''$. 
	Finally, we delete the edge $e=\{v_1,v_2\}$. 
	
	This completes the construction; we call the resulting graph $\widehat{G}$.  
	For each edge $e=\{v_1,v_2\}\in S$, we denote the induced subgraph of $\widehat{G}$ on the vertices $V(G_e')\cup V(G_e'')\cup\{u_e',u_e'',v_1,v_2\}$ by $H_e$. 
	We now examine the induced subgraph $H_e$ as a gadget that depends on $G$ and substitutes $e$. 
	
	In the following paragraph, we prove an essential property of the graph $H_e$, namely that it behaves exactly like the deleted edge $e$ as far as $(\chi(G)+1)$-colorability is concerned, whereas it acts like a nonedge with regards to $(\chi(G)+2)$-colorability. 
	
	Assume that $\chi(G)+1$ colors are available. First, let $v_1$ and $v_2$ be colored by two different ones of them. We can then extend this to a $(\chi(G)+1)$-coloring of $H_e$ in the following way. We assign to $u_e'$ and $u_e''$ the color of $v_1$, choose an arbitrary coloring of $G$ using the remaining $\chi(G)$ colors, and then assign it to both $G_e'$ and $G_e''$. 
	Now, pick instead an arbitrary color and let both $v_1$ and $v_2$ be colored by it. While this can be extended to a $(\chi(G)+2)$-coloring of $H_e$ immediately -- assign a second color to $u_e'$ and $u_e''$ and then color $G_e'$ and $G_e''$ with the remaining $\chi(G)$ colors -- it is impossible to extend it to a $(\chi(G)+1)$-coloring of $H_e$ as we show now. Seeking contradiction, assume there were such a coloring. It must assign to $u_e''$ a color different from the one color assigned to both $v_1$ and $v_2$. Since the vertices of $G_e'$ are all adjacent to both $v_1$ and $u_e''$, they must be colored with the remaining $\chi(G)-1$ colors, yielding the desired contradiction to $\chi(G_e')=\chi(G)$. 
	
	Returning to the entire graph, it remains to show that $\widehat{G}$ has the three stated properties.  
	\begin{enumerate}
		\item 
		Let a $\chi(G)$-coloring of $G$ be given. We describe how to extend it to a $(\chi(G)+2)$-coloring of $\widehat{G}$. First, we consider the coloring that results from the given coloring of $G$ by swapping out the two colors assigned to $v_1$ and $v_2$ for two new ones wherever they occur. We then color, for every $e\in S$, the copies $G_e'$ and $G_e''$ of $G$ according to the original $\chi(G)$-coloring. 
		We now assign to $w_1'$, $w_1''$, and, for every $e\in S$, to $u_e'$ and $u_e''$ the former color of $v_1$ and to $w_2'$ and $w_2''$ the former color of $v_2$. We can check that this yields a valid $(\chi(G)+2$)-coloring of $\widehat{G}$, 
		proving that $\chi(\widehat{G})\le \chi(G)+2$. 
		
		For the reverse inequality, assume by contradiction that $\widehat{G}$ has a $(\chi(G)+1)$-coloring. We observe two properties of this coloring. On the one hand, it uses at most $\chi(G)-1$ colors for the vertices $V(G)$ since they are all adjacent to the $2$-clique $\{w_1',w_2'\}$. On the other hand, the restriction to $V(H_e)$, for any $e\in S$, is a $\chi(G)+1$ coloring of $H_e$, which implies -- by the edge-like behavior of $H_e$ under this circumstance proved above -- that $v_1$ and $v_2$ are assigned different colors. Combining these two insights, we see that the restriction of our coloring to $V(G)$ is a $(\chi(G)-1)$-coloring of $G$, yielding the desired contradiction. 
		
		\item 
		First note that all edges in $S$ are deleted during the described construction of $\widehat{G}$, implying $E(\widehat{G})\setminus E(G)=E(\widehat{G})\setminus (E(G)\setminus S)$. 
		It now suffices to check that the vertices $V(\widehat{G})\setminus V(G)$ are all replicated or replicas and thus stable. Each edge $e$ in $E(\widehat{G})\setminus E(G)$ is adjacent to such a stable vertex and therefore stable itself by Observation~\ref{obs:StabilityFromVertices}. 
		\item 
		Note that the described construction does not commute with the deletion of an arbitrary edge $e\in E(G)$. (That is, denoting the construction by $f\colon G\mapsto \widehat{G}$, we do not have $f(G-e)=f(G)-e$.) This stands in contrast to the situation with the other stabilizing constructions in this paper -- where commutativity of the construction immediately yields the corresponding property -- necessitating an independent proof in this case. 
		
		Let $e'\in E(G)\setminus S$. 
		Assume first that $e$ is critical in $G$. Then there is a $(\chi(G)-1)$-coloring of $G-e'$. 
		We can extend it to a $(\chi(G)+1)$-coloring of $\widehat G-e'$ by assigning one of the two new colors to $w_1'$ and $w_2'$, the other one to $w_2'$ and $w_2''$, and then use the already proven fact that, for each $e\in S$, the graph $H_e$ can be colored with $\chi(G)+1$ colors whenever two different colors are prescribed for the two endpoints of $e$. This proves $\chi(\widehat G-e')\neq\chi(\widehat{G})$ and thus the criticality of $e'$ in $\widehat{G}$. 
		
		Now we start with the assumption that $e'$ is critical in $\widehat{G}$. By the first property of Lemma~\ref{lem:StabilizingGadget} we have $\chi(\widehat{G}-e')=\chi(G)+1$.
		Pick an arbitrary $(\chi(G)+1)$-coloring of $\widehat{G}-e'$. 
		Since the induced 4-cycle on $\{w_1',w_2',w_1'',w_2''\}$ has chromatic number 2 and is joined to the induced subgraph on $V(G)$, the induced coloring on $G-e'$ uses at most $\chi(G)-1$ of these colors, implying $\chi(G-e')\le\chi(G)-1$ and thus the criticality of $e'$ in $G$.  
		We deduce via the contrapositive that $e'$ is stable in $G$ if and only if it is stable in $\widehat{G}$. 
	\end{enumerate}
This concludes the proof of Lemma~\ref{lem:StabilizingGadget}. 
\end{proof}

\section{Proof of Corollary~\ref{cor:and-for-stability}}\label{app:and-for-stability}
We restate and prove Corollary~\ref{cor:and-for-stability}.
\setcounter{tempcounter}{\value{theorem}}
\setcounterref{theorem}{cor:and-for-stability}
\addtocounter{theorem}{-1}
\begin{corollary}
	Mapping $k$ graphs $G_1,\ldots,G_k$ to $G_1+\cdots+G_k$ with all join edges stabilized using the construction from Lemma~\ref{lem:StabilizingGadget} is an \ANDomega function for \Stability.
\end{corollary}
\setcounter{theorem}{\value{tempcounter}}
\begin{proof}
	We know that $\chi(G_1+\cdots+G_k)=\chi(G_1)+\cdots+\chi(G_k)$. This implies that, for any $i\in\{1,\ldots,k\}$, an edge $e\in E(G_i)$ is stable in $G_i$ exactly if it is stable in $G_1+\cdots+G_k$. The graph with all join edges stabilized is thus stable exactly if all graphs $G_1,\ldots,G_k$ are. 
	Note that the more complicated formulation of Lemma~\ref{lem:StabilizingGadget} that allows for the stabilization of arbitrary subsets of edges rather than just a single chosen edge is crucial for the derivation of this corollary. For since the construction of Lemma~\ref{lem:StabilizingGadget} more than doubles the number of vertices -- even if $S$ is only a singleton -- at most a logarithmic number of iterated applications are possible in polynomial time. 
\end{proof}

\section{Proof of Theorem~\ref{thm:beta-vertex-stability}}\label{app:beta-vertex-stability}

We restate and prove Theorem~\ref{thm:beta-vertex-stability}.
\setcounter{tempcounter}{\value{theorem}}
\setcounterref{theorem}{thm:beta-vertex-stability}
\addtocounter{theorem}{-1}
\begin{theorem}
	Only the empty graphs are $\beta$-vertex-stable.
\end{theorem} 
\begin{proof}
	Let $G$ be a graph. If $G$ is empty, it is $\beta$-vertex-stable since the empty set is a minimum vertex cover for both $G$ and $G-v$ for every $v\in V(G)$. 
	If $G$ has an edge $\{u,v\}$, every vertex cover contains either $u$ or $v$ or both. Consider any optimal vertex cover $X$ of $G$ and assume, without loss of generality, that $v\in X$. Then $v$ is a critical vertex since $X\setminus\{v\}$ is a vertex cover of size $\card{X}-1$ of $G-v$. 
\end{proof}
\setcounter{theorem}{\value{tempcounter}}

\section{Proof of Lemma~\ref{lem:BetaStabilizingGadget}}
\label{app:BetaStabilizingGadgetProof}
We restate and prove Lemma~\ref{lem:BetaStabilizingGadget}.
\setcounter{tempcounter}{\value{theorem}}
\setcounterref{theorem}{lem:BetaStabilizingGadget}
\addtocounter{theorem}{-1}
\begin{lemma}
	Let $G$ be a graph and $\{v_1,v_2\}\in E(G)$ one of its edges. Create from $G$ a new graph $G'$ by replacing the edge $\{v_1,v_2\}$ by the gadget that consists of four new vertices $u_1$, $u_2$, $u_3$, and $u_4$ with edges $\{u_1,u_2\}$, $\{u_2,u_3\}$, $\{u_3,u_4\}$, and $\{u_4,u_1\}$ (i.e., a new rectangle) and additionally the edges $\{v_1,u_1\}$, $\{v_1,u_3\}$, $\{v_2,u_2\}$, and $\{v_2,u_4\}$. 
	(This gadget is displayed in Figure~\ref{fig:BetaStabilizingGadgetB}.)
	Then we have $\beta(G')	= \beta(G)+2$, all edges of the gadget are stable in $G'$, and the remaining edges are stable in $G'$ if and only if they are stable in $G$.
\end{lemma} 
\setcounter{theorem}{\value{tempcounter}}

\begin{figure}[htbp]
	\begin{center}
\begin{subfigure}[c]{0.49\textwidth}
	\begin{tikzpicture}[
	xscale=1,yscale=1,
	vertex/.style={draw,circle, inner sep=.0pt, minimum size=0.53cm},font=\scriptsize]

	\node[vertex] (v1) at (-1.41,0) {$v_1$};
	\node[vertex] (v2) at (1.41,0) {$v_2$};
	
	\draw (v1) -- (v2);
	\draw[white] (0,0) ellipse (1cm and 1.5cm); 

	\draw[dashed] (0,0) ellipse (2cm and 1cm);
	\clip (0,0) ellipse (2cm and 1cm);
	
	\draw[] (v1) -- (-2,.5);
	\draw[] (v1) -- (-2.2,-.2);
	\draw[] (v1) -- (-1.2,-1.5);
	
	\draw[] (v2) -- (2,1.3);
	\draw[] (v2) -- (2,-.3);
	\end{tikzpicture}
\captionsetup{width=1\linewidth}
\subcaption{Example of the relevant section of $G$.}
\end{subfigure}
\begin{subfigure}[c]{0.49\textwidth}
	\begin{tikzpicture}[
xscale=1,yscale=1,
vertex/.style={draw,circle, inner sep=.0pt, minimum size=0.53cm},font=\scriptsize]
	\node[vertex] (v1) at (-2.41,0) {$v_1$};
	\node[vertex] (u1) at (-1,1) {$u_1$};
	\node[vertex] (u2) at (1,1) {$u_2$};
	\node[vertex] (u3) at (-1,-1) {$u_3$};
	\node[vertex] (u4) at (1,-1) {$u_4$};
	\node[vertex] (v2) at (2.41,0) {$v_2$};

	\draw (v1) -- (u1) -- (u2) -- (v2);
	\draw (v1) -- (u3) -- (u4) -- (v2);
	\draw (u1) -- (u4);
	\draw (u2) -- (u3);
	
	\draw[dashed] (0,0) ellipse (3cm and 1.5cm);
	\clip (0,0) ellipse (3cm and 1.5cm);
	
	\draw[] (v1) -- (-3,.5);
	\draw[] (v1) -- (-3.2,-.2);
	\draw[] (v1) -- (-2.2,-1.5);
	
	\draw[] (v2) -- (3,1.3);
	\draw[] (v2) -- (3,-.3);
	\end{tikzpicture}
\captionsetup{width=1\linewidth}
\subcaption{The same section in $G'$ after the substitution.}
\label{fig:BetaStabilizingGadgetB}
\end{subfigure}
\caption{Illustration of the substitution of an edge $\{v_1,v_2\}$
  by the gadget mentioned in
  Lemma~\ref{lem:BetaStabilizingGadget}.
We remark in passing that the gadget used here is the smallest one with the desired properties.
}
	\end{center}
\end{figure}
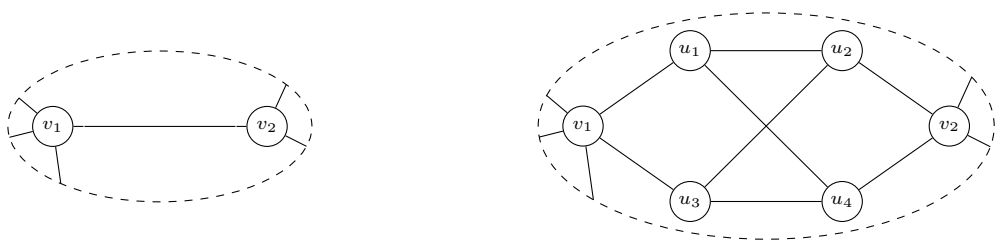

\begin{proof}
	Let $X$ be a vertex cover for $G$. Due to $\{v_1,v_2\}\in E(G)$, we have $v_1\in X$ or $v_2\in X$. We obtain a vertex cover for $G'$ by adding $u_2$ and $u_4$ to $X$ in the first case and $u_1$ and $u_3$ in the second case.  This shows $\beta(G') \le \beta(G)+2$. For the reverse inequality observe that out of the four vertices $\{u_1,u_2,u_3,u_4\}$ inducing a $4$-cycle every vertex cover for $G'$ contains at least two and removing all of them leaves us with a vertex cover for $G$. 
	
	Now, let $e$ be an arbitrary edge of $G$ other than $\{v_1,v_2\}$. We can check that the argument above still holds true for $G-e$ and $G'-e$ instead of $G$ and $G'$, and hence $\beta(G'-e) = \beta(G-e)+2$. 
	It follows that $\beta(G)-\beta(G-e)=\beta(G')-\beta(G'-e)$; 
	that is, $e$ is stable in $G$ exactly if it is stable in $G'$. 
	
	Finally, we prove $\beta$-stability for all the gadget edges. Let $e$ be such an edge. It suffices to show that, for every vertex cover of $G'-e$, there is a vertex cover of the same size for $G'$. 
	We show this for $e=\{u_1,u_2\}$ and $e=\{v_1,u_1\}$; the remaining cases follow immediately by symmetry. 
	The following argumentation is illustrated in
	Figure~\ref{fig:BetaStabilizingGadgetProof}. 
	Let $X$ be a vertex cover of $G'-\{u_1,u_2\}$. If $X$ contains $u_1$ or $u_2$, it is a vertex cover for $G'$ as well. Otherwise, we have $\{v_1,u_3,u_4,v_2\}\subseteq X$ since the edges incident to $u_1$ and $u_2$ need to be covered -- see Figure~\ref{fig:BetaStabilizingGadgetProof}a -- and $\{u_2\}\cup X\setminus\{u_3\}$ is a vertex cover for $G'$; see Figure~\ref{fig:BetaStabilizingGadgetProof}b. 
	Analogously, let $Y$ be vertex cover of $G'-\{v_1,u_1\}$. If $Y$ contains $v_1$ or $u_1$, then it is a vertex cover of $G'$ as well. Otherwise, $Y$ contains the vertices $\{u_2,u_3,u_4\}$ since they are neighbors of either $v_1$ or $u_1$; see Figure~\ref{fig:BetaStabilizingGadgetProof}c. Then, $\{u_1,v_2\}\cup Y\setminus\{u_2,u_4\}$ is a vertex cover for $G'$ that is either -- see Figure~\ref{fig:BetaStabilizingGadgetProof}d -- of the same size as $Y$ or, if $v_2\in Y$, smaller by one. 
\end{proof}

\setcounter{tempcounter}{\value{figure}}
\begin{figure}[htbp]
	\begin{center}
		\begin{subfigure}[c]{0.49\textwidth}
			\begin{tikzpicture}[xscale=1,yscale=1,
			vertex/.style={draw,circle, inner sep=.0pt, minimum size=0.53cm},font=\scriptsize]
			\node[vertex,fill=black!15] (v1) at (-2.41,0) {$v_1$};
			\node[vertex] (u1) at (-1,1) {$u_1$};
			\node[vertex] (u2) at (1,1) {$u_2$};
			\node[vertex,fill=black!15] (u3) at (-1,-1) {$u_3$};
			\node[vertex,fill=black!15] (u4) at (1,-1) {$u_4$};
			\node[vertex,fill=black!15] (v2) at (2.41,0) {$v_2$};
			
			\draw (v1) -- (u1);
			\draw (u2) -- (v2);
			\draw (v1) -- (u3) -- (u4) -- (v2);
			\draw (u1) -- (u4);
			\draw (u2) -- (u3);
			\end{tikzpicture}
			\captionsetup{width=1\linewidth}
			\subcaption{A vertex cover $X$ for $G'-\{u_1,u_2\}$.\\ Only the gadget part is shown.}
		\end{subfigure}
		\begin{subfigure}[c]{0.49\textwidth}
			\begin{tikzpicture}[xscale=1,yscale=1,
			vertex/.style={draw,circle, inner sep=.0pt, minimum size=0.53cm},font=\scriptsize]
			\node[vertex,fill=black!15] (v1) at (-2.41,0) {$v_1$};
			\node[vertex] (u1) at (-1,1) {$u_1$};
			\node[vertex,fill=black!15] (u2) at (1,1) {$u_2$};
			\node[vertex] (u3) at (-1,-1) {$u_3$};
			\node[vertex,fill=black!15] (u4) at (1,-1) {$u_4$};
			\node[vertex,fill=black!15] (v2) at (2.41,0) {$v_2$};
			
			\draw (v1) -- (u1) -- (u2) -- (v2);
			\draw (v1) -- (u3) -- (u4) -- (v2);
			\draw (u1) -- (u4);
			\draw (u2) -- (u3);
			\end{tikzpicture}
			\captionsetup{width=1\linewidth}
			\subcaption{Removing $u_3$ and adding $u_2$ yields\\ a vertex cover of the same size for $G'$.}
		\end{subfigure}
		\begin{subfigure}[c]{0.49\textwidth}
			\begin{tikzpicture}[xscale=1,yscale=1,
			vertex/.style={draw,circle, inner sep=.0pt, minimum size=0.53cm},font=\scriptsize]
			\node[vertex] (v1) at (-2.41,0) {$v_1$};
			\node[vertex] (u1) at (-1,1) {$u_1$};
			\node[vertex,fill=black!15] (u2) at (1,1) {$u_2$};
			\node[vertex,fill=black!15] (u3) at (-1,-1) {$u_3$};
			\node[vertex,fill=black!15] (u4) at (1,-1) {$u_4$};
			\node[vertex] (v2) at (2.41,0) {$v_2$};
			
			\draw (u1) -- (u2) -- (v2);
			\draw (v1) -- (u3) -- (u4) -- (v2);
			\draw (u1) -- (u4);
			\draw (u2) -- (u3);
			\end{tikzpicture}
			\captionsetup{width=1\linewidth}
			\subcaption{A vertex cover $Y$ for $G'-\{v_1,u_1\}$.}
		\end{subfigure}
		\begin{subfigure}[c]{0.49\textwidth}
			\begin{tikzpicture}[xscale=1,yscale=1,
			vertex/.style={draw,circle, inner sep=.0pt, minimum size=0.53cm},font=\scriptsize]
			\node[vertex] (v1) at (-2.41,0) {$v_1$};
			\node[vertex,fill=black!15] (u1) at (-1,1) {$u_1$};
			\node[vertex] (u2) at (1,1) {$u_2$};
			\node[vertex,fill=black!15] (u3) at (-1,-1) {$u_3$};
			\node[vertex] (u4) at (1,-1) {$u_4$};
			\node[vertex,fill=black!15] (v2) at (2.41,0) {$v_2$};
			
			\draw (v1) -- (u1) -- (u2) -- (v2);
			\draw (v1) -- (u3) -- (u4) -- (v2);
			\draw (u1) -- (u4);
			\draw (u2) -- (u3);
			\end{tikzpicture}
			\captionsetup{width=1\linewidth}
			\subcaption{A vertex cover of the same size for $G'$.}
		\end{subfigure}
	\end{center}
	\caption{An illustration of the proof of
          Lemma~\ref{lem:BetaStabilizingGadget}: All edges of the gadget
          are $\beta$-stable in $G'$.}\label{fig:BetaStabilizingGadgetProof}
\end{figure}
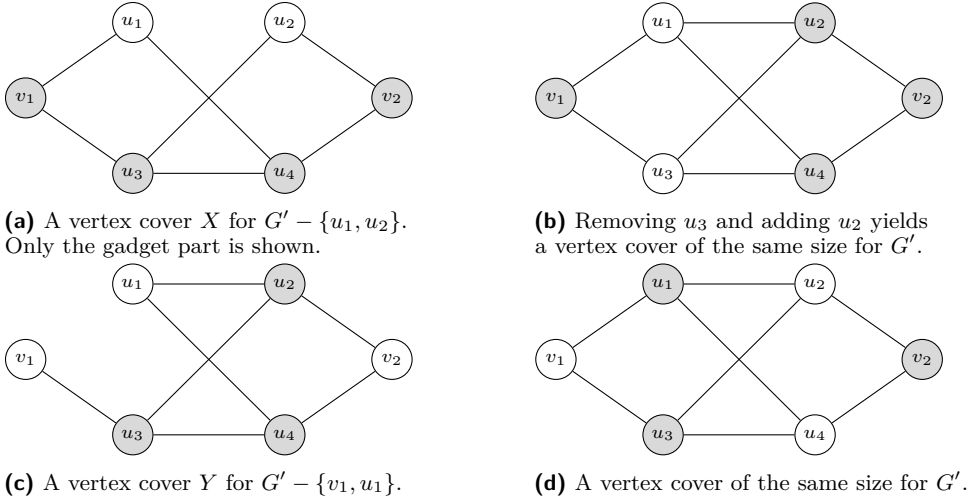

\section{Proof of Theorem~\ref{thm:v-unfrozenness}}
\label{app:v-unfrozenness}
We restate and prove Theorem~\ref{thm:v-unfrozenness}.
\setcounter{tempcounter}{\value{theorem}}
\setcounterref{theorem}{thm:v-unfrozenness}
\addtocounter{theorem}{-1}
\begin{theorem}
	There is no vertex-unfrozen graph and only one $\beta$-vertex-unfrozen graph, namely the null graph (i.e., the graph with the empty vertex set).
\end{theorem}
\begin{proof}
	Adding a universal vertex increases both $\chi$ and $\beta$ by exactly one, with the exception of the null graph $K_0$, for which $\chi(K_0)=\beta(K_0)=0$ but $\chi(K_1)=1\neq\beta(K_1)=0$.  
\end{proof}
\setcounter{theorem}{\value{tempcounter}}

\section{Proof of Theorem~\ref{thm:beta-unfrozenness}}
\label{app:beta-unfrozenness}
We restate and prove Theorem~\ref{thm:beta-unfrozenness}.
\setcounter{tempcounter}{\value{theorem}}
\setcounterref{theorem}{thm:beta-unfrozenness}
\addtocounter{theorem}{-1}
\begin{theorem}
	Determining whether a graph is $\beta$-unfrozen is \ThetaTwo-complete.
\end{theorem}
\setcounter{theorem}{\value{tempcounter}}
\begin{proof}
	The upper bound is again immediate. For hardness, we describe a reduction from $\alternativetextsc{CompareVertexCover}=\{(G,H)\in\mathcal{G}^2\mid \beta(G) \le  \beta(H)\}$ -- known to be \ThetaTwo-hard~~\cite[Thm.~12]{spa-vog:c:theta-two-classic} -- to $\beta$-\Unfrozenness. 
	Let $G$ and $H$ be the two given graphs. 
	We show how to construct in polynomial time from $G$ and $H$ a graph $J$ that is $\beta$-unfrozen if and only if $\beta(G) \le  \beta(H)$. 
	Denote $g = \card{V(G)}$ and $h=\card{V(H)}$.
	Assume without loss of generality that $g >  1$. 
	Let $G' = (G \cup I_h) + (G \cup I_h)$.
	Note that
	$G'$ is $\beta$-unfrozen, with
	$\beta(G') = \beta(G) + g + h$, and that
	$\card{V(G')} = 2(g + h)$.
	Let $H' = (H + K_{g+h}) \cup I_g$ (i.e., join a $(g+h)$-clique to $H$ and then add $g$ isolated vertices). 
	Note that $H'$ is
	not $\beta$-unfrozen -- due to the $\beta$-frozen edges that can be added between any two of the $g \ge  2$ vertices of $I_g$ -- with
	$\beta(H')$ = $\beta(H + K_{g+h})$ = $\beta(H) + g + h$ and that
	$\card{V(H')} = (h + (g + h)) + g = 2(g + h)$.
	Let $c = 2(g + h)$.
	We conclude that  
	$G'$ is $\beta$-unfrozen, that 
	$\beta(G') \le  \beta(H') \Longleftrightarrow \beta(G) \le \beta(H)$, that
	$\card{V(G')} = \card{V(H')} = c$, and that
	$H'$ is $\beta$-frozen.
	Now, let $J = G' + H'$. Clearly, $J$ can be constructed in polynomial time.
	We will prove that $J$ is $\beta$-unfrozen if and only if $\beta(G) \le \beta(H)$.
	For both directions, note that due to $\card{V(G')} = \card{V(H')} = c$, we have $\beta(J) = \min\{\beta(G'),\beta(H')\} + c$.
	
	First, assume that $\beta(G) \le \beta(H)$. Then $\beta(G') \le \beta(H')$ and thus $\beta(J) = \beta(G') + c$. 
	We prove that all nonedges $e\in \coE(J)$ are $\beta$-unfrozen in $J$.
	Such an $e$ is either adjacent to two vertices of $G'$ or to two vertices of $H'$.
	For the first case, note that $\beta(G'+e) = \beta(G')$ since $G'$ is $\beta$-unfrozen.
	Thus we have $\beta(J+e) = \min\{\beta(G'+e), \beta(H')\} + c
	= \beta(J)$.
	In the second case we have $\beta(J+e) = \min\{\beta(G'), \beta(H' + e)\} + c = \beta(G') + c = \beta(J)$, where the second equality follows from $\beta(G') \le \beta(H') \le \beta(H'+e)$.
	Thus $J$ is $\beta$-unfrozen if $\beta(G) \le \beta(H)$.
	
	For the converse, assume that $\beta(G) > \beta(H)$, implying $\beta(G') > \beta(H')$ and thus $\beta(J) = \beta(H') + c$.
	Since $H'$ is $\beta$-frozen, there is a $\beta$-frozen nonedge $e$ that can be added to $H'$, yielding $\beta(H'+e) = \beta(H') + 1 \le  \beta(G)$. Hence, we have $\beta(J+e) = \min\{\beta(G'), \beta(H'+e)\} + c = \beta(H') + 1 + c > \beta(J)$. This shows that $e$ is $\beta$-frozen for $J$ as well, concluding the proof. 
\end{proof}
\setcounter{theorem}{\value{tempcounter}}

\section{Proof of Theorem~\ref{thm:unfrozenness}}
\label{app:unfrozenness}
We restate and prove Theorem~\ref{thm:unfrozenness}.
\setcounter{tempcounter}{\value{theorem}}
\setcounterref{theorem}{thm:unfrozenness}
\addtocounter{theorem}{-1}
\begin{theorem}
	Assume that there are polynomial-time computable functions  $f\colon \mathcal{G}\to \mathcal{G}$ and $g\colon \mathcal{G}\to \Z$ such that for any graph $G$ we have that $f(G)$ is unfrozen and $\chi(f(G))=\chi(G)+g(G)$. 
	Then \Unfrozenness is \ThetaTwo-complete.
\end{theorem} 
\begin{proof}
	As usual, membership in \ThetaTwo is immediate. For the lower bound, assume the existence of functions $f$ and $g$ as described in the theorem. 
	We reduce from the problem $\{(G,H)\mid \chi(G)\le \chi(H)\}$, whose \ThetaTwo-hardness is stated in Theorem~\ref{thm:CompareColorability}. 
	
	Let two graphs $G$ and $H$ be given. 
	Our intermediate goal is to construct two graphs $G''$ and $H''$ such that only the latter is unfrozen and $\chi(G)\le \chi(H)\ \Longleftrightarrow\ \chi(G'')<\chi(H'')$. 
	Let $G'=G+ I_2$ and $H'=f(H)$. Note that $H'$ is unfrozen with $\chi(H')=\chi(H)+g(H)$, while $G'$ is not unfrozen -- due to the frozen nonedge that can be added to $I_2$ -- and $\chi(G')=\chi(G)+1$.
	Let $G''=G'+K_{\max\{0,g(H)-1\}}$ and $H''= H'+K_{1+\max\{1-g(H),0\}}$.  
	Then $H''$ is unfrozen with $\chi(H'')=\chi(H)+g(H)+1+\max\{1-g(H),0\}=\chi(H)+\max\{1,g(H)\}+1$, while $G''$ has a frozen edge with $\chi(G'')=\chi(G)+1+\max\{0,g(H)-1\}=\chi(G)+\max\{1,g(H)\}$. 
	Moreover, we have $\chi(G'')-\chi(G)=\max\{1,g(H)\}=\chi(H'')-\chi(H)-1$ and hence $\chi(G)\le \chi(H)\ \Longleftrightarrow\ \chi(G'')<\chi(H'')$, as desired. 
	
	Now let $U = G'' \cup H''$ be the disjoint union of the two constructed graphs. 
	We clearly have $\chi(U) = \max\{\chi(G''),\chi(H'')\}$. 
	Moreover, every nonedge between an arbitrary vertex $v\in V(G'')$ and an arbitrary vertex $w\in V(H'')$ is unfrozen: 
	Let an arbitrary optimal coloring of $U=G''\cup H''$ be given. If it assigns $v$ and $w$ different colors, we are done; otherwise, swap the two distinct colors of $v$ and $w$ for all vertices in $V(G'')$. 
	Recalling that $H''$ is unfrozen, we can conclude that $U$ is unfrozen if and only if all nonedges $e\in\coE(G'')=\coE(G')$ are unfrozen. 
	
	Assume first that $\chi(G)\le \chi(H)$; that is, $\chi(G'')<\chi(H'')$. For every $e\in\coE(G'')$, we then have $\chi(U+e)=\chi((G''+e)\cup H'')\le \max\{\chi(G'')+1,\chi(H'')\}=\chi(H'')=\chi(U)$. Therefore, $U$ is unfrozen. 
	Assume now that $\chi(G)>\chi(H)$; that is, $\chi(G'')\ge\chi(H'')$. Recall that $G''$ has a frozen edge $e'$. It follows that $\chi(G''+e')=\chi(G'')+1>\chi(H'')$ and thus  $\chi(U+e')=\max\{\chi(G''+e'),\chi(H'')\}=\chi(G'')+1>\max\{\chi(G''),\chi(H'')\}=\chi(U)$. Therefore, $e'$ is also a frozen edge of $U$, which is thus not unfrozen. This concludes the proof. 	
\end{proof}
\setcounter{theorem}{\value{tempcounter}}

\section{Proof of Theorem~\ref{thm:CompareColorability}}
\label{app:CompareColorability}
We will restate and prove Theorem~\ref{thm:CompareColorability} after stating in Lemmas~\ref{lem:equal} and~\ref{lem:compare} two sufficient criteria for \ThetaTwo-hardness. 
Both of them are consequences of the Wagner's criterion~\cite[Thm.~5.2]{wag:j:more-on-bh}, which we state in Lemma~\ref{lem:odd}. 
We assume without loss of generality that all problems are encoded over the same finite alphabet $\Sigma$.

\begin{lemma}\label{lem:odd}
	A problem $A$ is \ThetaTwo-hard if the following condition is satisfied:
	
	There is an \NP-complete problem $D$ and there is a polynomial-time computable function $f\colon \bigcup_{k=1}^\infty(\Sigma^\ast)^{2k}\to \Sigma^\ast$ 
	such that for every $k\in\N\setminus\{0\}$ and for all $x_1,\dots,x_{2k}\in\Sigma^\ast$ with $x_1\in D \Leftarrow \dots \Leftarrow x_{2k}\in D$ we have 
	\[f(x_1,\dots,x_{2k})\in A\quad \Longleftrightarrow\quad  \card{\{x_1,\dots,x_{2k}\}\cap D}\text{ is odd}.\]
\end{lemma}

The following two lemmas are identical, except for the last line, where once we have an equality and once a nonstrict inequality. 

\begin{lemma}\label{lem:equal}
	A problem $A$ is \ThetaTwo-hard if the following condition is satisfied:
	
	There are \NP-complete problems $D_1$ and $D_2$ and a polynomial-time computable function $g\colon \bigcup_{k=1}^\infty(\Sigma^\ast)^{2k}\to \Sigma^\ast$ 
	such that for every $k\in\N\setminus\{0\}$ and for all $y_1,\dots,y_k,z_1,\dots,z_k\in\Sigma^\ast$ with $y_1\in D_1 \Leftarrow \dots \Leftarrow y_k\in D$ and $z_1\in D_2 \Leftarrow \dots \Leftarrow z_k\in D_2$ we have 
	\[g(y_1,\dots,y_k,z_1,\dots,z_k)\in A\quad \Longleftrightarrow\quad  \card{\{y_1,\dots,y_k\}\cap D_1}= \card{\{z_1,\dots,z_k\}\cap D_2}.\]
\end{lemma}

\begin{lemma}\label{lem:compare}
	A problem $A$ is \ThetaTwo-hard if the following condition is satisfied:
	
	There are \NP-complete problems $D_1$ and $D_2$ and a polynomial-time computable function $g\colon \bigcup_{k=1}^\infty(\Sigma^\ast)^{2k}\to \Sigma^\ast$ 
	such that for every $k\in\N\setminus\{0\}$ and for all $y_1,\dots,y_k,z_1,\dots,z_k\in\Sigma^\ast$ with $y_1\in D_1 \Leftarrow \dots \Leftarrow y_k\in D$ and $z_1\in D_2 \Leftarrow \dots \Leftarrow z_k\in D_2$ we have 
	\[g(y_1,\dots,y_k,z_1,\dots,z_k)\in A\quad \Longleftrightarrow\quad  \card{\{y_1,\dots,y_k\}\cap D_1}\le \card{\{z_1,\dots,z_k\}\cap D_2}.\]
\end{lemma}

We point out again that Lemma~\ref{lem:compare} is the corrected version of a slightly flawed lemma statement in the paper by Spakowski and Vogel~\cite[Lem.~9]{spa-vog:c:theta-two-classic}. Since said paper does not apply the problematic lemma anywhere, all other results derived in it remain valid.	

\begin{proof}[Proof of Lemmas~\ref{lem:equal} and~\ref{lem:compare}]
It suffices to make the following five observations. 

First, \ThetaTwo is closed under complementation since we can just invert the output of any algorithm witnessing the membership in $\P^\NP_\|$. Therefore, Lemma~\ref{lem:odd} remains true if we replace, on the one hand, the language $A$ by its complement $\bar{A}=\Sigma^\ast\setminus A$ in both occurrences and, on the other hand, ``odd'' by ``even.'' 

Second, due to the assumption $x_1\in D \Leftarrow x_2\in D \Leftarrow \dots \Leftarrow  x_{2k}\in D$, the following three conditions are now equivalent within the modified  Lemma~\ref{lem:odd}. 

\begin{enumerate}
	\item $\card{\{x_1,\dots,x_{2k}\}\cap D}\text{ is even}$,
	\item $\card{\{x_1,x_3,\dots,x_{2k-1}\}\cap D}=\card{\{x_2,x_4\dots,x_{2k}\}\cap D}$, and
	\item $\card{\{x_1,x_3,\dots,x_{2k-1}\}\cap D}\le  \card{\{x_2,x_4\dots,x_{2k}\}\cap D}$. 
\end{enumerate} 

Third, given two arbitrary \NP-complete problems $D_1$ and $D_2$, there are polynomial-time many-one reduction $h_1$ and $h_2$ from $D$ to $D_1$ and $D_2$, respectively. 
Letting 
\begin{align*}
y_1&=h_1(x_1),\ y_2=h_1(x_3),\ \dots,\ h_1(x_{2k-1})\text{ and}\\
z_1&=h_2(x_2),\ z_2=h_2(x_4),\ \dots,\ h_2(x_{2k}), 
\end{align*}
we have 
\begin{align*}
\card{\{x_1,x_3,\dots,x_{2k-1}\}\cap D}&=\card{\{y_1,y_2,\dots,y_k\}\cap D_1}\text{ and }\\
\card{\{x_2,x_4,\dots,x_{2k}\}\cap D}&=\card{\{z_1,z_2,\dots,z_k\}\cap D_2}.
\end{align*}

Forth, given a polynomial-time computable function $g\colon \bigcup_{k=1}^\infty(\Sigma^\ast)^{2k}\to \Sigma^\ast$, we obtain another such function $f$ that satisfies $f(x_1,\dots,x_{2k})\in A \Longleftrightarrow g(y_1,\dots,y_k,z_1,\dots,z_k)\in A$ by simply defining $f(x_1,\dots,x_{2k})=g(y_1,\dots,y_k,z_1,\dots,z_k)$.

Finally, $x_1\in D \Leftarrow \dots \Leftarrow x_{2k}\in D$ implies both $x_1\in D \Leftarrow x_3\in D \Leftarrow \dots \Leftarrow x_{2k-1}\in D$ and $x_2\in D \Leftarrow x_4\in D \Leftarrow \dots \Leftarrow x_{2k}\in D$, 
which in turn implies $y_1\in D_1 \Leftarrow \dots \Leftarrow y_k\in D_1$ and $z_1\in D_2 \Leftarrow \dots \Leftarrow z_k\in D_2.$
\end{proof}

Lemma~\ref{lem:equal} provides for several equality problems the proofs of \ThetaTwo-hardness (and thus \ThetaTwo-completeness), which Wagner asserted~\cite[page 79]{wag:j:more-on-bh} without spelling out the straightforward proofs. 
In particular, \alternativetextsc{EqualIndependentSet}, \alternativetextsc{EqualVertexCover}, \alternativetextsc{EqualColorability}, and \alternativetextsc{EqualClique} -- which ask whether two graphs have the same graph number $\alpha$, $\beta$, $\chi$, and $\omega$, respectively -- and \alternativetextsc{EqualMaxSat} -- which asks whether two formulas in \ThreeCNF have the same maximal number of simultaneously satisfiable clauses -- are all \ThetaTwo-complete. This is seen by applying Proposition~\ref{lem:equal} to the proofs of the corresponding theorems by Wagner~\cite[Thms. 6.1, 6.2, and 6.3]{wag:j:more-on-bh}.

By applying Lemma~\ref{lem:compare} instead, we immediately obtain \ThetaTwo-completeness for the comparison problems \alternativetextsc{CompareIndependentSet}, \alternativetextsc{CompareVertexCover}, \alternativetextsc{CompareColorability}, \alternativetextsc{CompareClique}, and \alternativetextsc{CompareMaxSat}. 

For all but \alternativetextsc{CompareColorability}, \ThetaTwo-hardness was also proved by Spakowski and Vogel~\cite[Thms.~2, 12 and 13]{spa-vog:c:theta-two-classic}. 
We now show how to apply concretely Lemma~\ref{lem:compare} to obtain the \ThetaTwo-hardness of \alternativetextsc{CompareColorability}. 

\setcounter{tempcounter}{\value{theorem}}
\setcounterref{theorem}{thm:CompareColorability}
\addtocounter{theorem}{-1}
\begin{theorem}
	$\alternativetextsc{CompareColorability}=\{(G,H)\in\mathcal{G}^2\mid \chi(G)\le\chi(H)\}$ is \ThetaTwo-hard. 
\end{theorem}
\setcounter{theorem}{\value{tempcounter}}
\begin{proof}
	This proof is modeled after the one for Wagner's Theorem 6.3. 
	We apply Lemma~\ref{lem:compare} with $A=\alternativetextsc{CompareColorability}$ and $D_1=D_2=\ThreeSat$. 
	
	Let $k\in\N\setminus\{0\}$ and let $2k$ formulas $\Phi_1,\dots,\Phi_k,\Psi_1,\dots,\Psi_k\in\ThreeCNF$
	satisfying 
	\begin{align*}
	&\Phi_1\in\ThreeSat\Leftarrow\dots\Leftarrow\Phi_k\in\ThreeSat\text{ and }\\
	&\Psi_1\in\ThreeSat\Leftarrow\dots\Leftarrow\Psi_k\in\ThreeSat
	\end{align*} 
	be given. 
	Denote by $h$ the standard reduction from \ThreeSat to \ThreeCol by Garey et al.~\cite{gar-joh-sto:j:some-simplified-np-complete-graph-problems}; it maps a formula to a graph whose chromatic number is 3 if the formula is satisfiable and 4 otherwise. 
	Moreover, let    
	$G=h(\Phi_1)+\dots+h(\Phi_k)$ and $H=h(\Psi_1)+\dots+h(\Psi_k)$, where $+$ denotes the graph join. 
	It follows that 
	\begin{align*}
	\chi(G)=\sum_{i=1}^k\chi(h(\Phi_i))&=k+\card{\{\Phi_1,\dots,\Phi_k\}\cap\ThreeSat}\quad\text{ and }\\
	\chi(H)=\sum_{i=1}^k\chi(h(\Psi_i))&=3k+\card{\{\Psi_1,\dots,\Psi_k\}\cap\ThreeSat}.
	\end{align*}
	Thus we have
	$\chi(G)\le\chi(H)$ 
	if and only if
	$\card{\{\Phi_1,\dots,\Phi_k\}\cap\ThreeSat}\le \card{\{\Psi_1,\dots,\Psi_k\}\cap\ThreeSat}$. 	
	The map $g\colon (\Phi_1,\dots,\Phi_k,\Psi_1,\dots,\Psi_k)\mapsto (G,H)$ therefore satisfies all requirements of Lemma~\ref{lem:compare}, which concludes the proof. 
	\end{proof}

\section{Proof of Theorem~\ref{thm:twowaystability}}
\label{app:twowaystability}
We restate and prove Theorem~\ref{thm:twowaystability}.
\setcounter{tempcounter}{\value{theorem}}
\setcounterref{theorem}{thm:twowaystability}
\addtocounter{theorem}{-1}
\begin{theorem}
	The problem \TwoWayStability is \ThetaTwo-complete.
\end{theorem}
\begin{proof}
	We have $\TwoWayStability = \Stability \cap \Unfrozenness$.
	The membership in \ThetaTwo is immediate.
	For \ThetaTwo-hardness, we show that the map $f(G)= G\cup G$ is a reduction from \Unfrozenness.
	First, $G\cup G$ is stable for any given graph $G$ since $\chi (G_1\cup G_2) = \max\{\chi(G_1) , \chi(G_2)\}$ for all graphs $G_1,G_2\in\mathcal{G}$.
	We conclude that $G\cup G$ is two-way-stable if and only if it is unfrozen.
	Moreover, $G\cup G$ is unfrozen if and only if $G$ is unfrozen: 
	A nonedge $e\in\coE(G)$ is unfrozen in $G$ exactly if it is unfrozen in $G\cup G$, again due to  $\chi (G_1\cup G_2) = \max\{\chi (G_1) , \chi(G_2)\}$.
	It remains to examine the nonedges that can be added to $G\cup G$ between the two copies of $G$. Let $\{v_1,v_2\}$ be such a nonedge. We prove that it is unfrozen.  Without loss of generality, assume that $G$ is nonempty, that is, $\chi (G)>1$. Given an optimal coloring for $G$, we obtain an optimal coloring for $G\cup G+\{v_1,v_2\}$ by coloring both copies according to the given coloring, just with the colors permuted appropriately for the second copy, that is, such that $v_2$ receives a color different from the one of $v_1$.
\end{proof}
\setcounter{theorem}{\value{tempcounter}}

\section{Proof of Lemma~\ref{lem:quadripartite-gadget}}
\label{app:quadripartite-gadget}
We restate and prove Lemma~\ref{lem:quadripartite-gadget}.

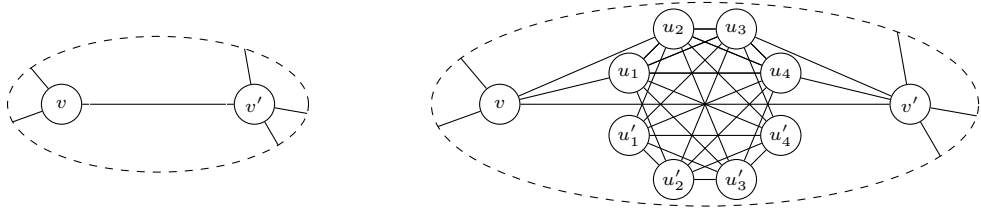
\begin{figure}[htbp]
	\begin{center}
		\begin{subfigure}[c]{0.39\textwidth}
			\begin{tikzpicture}[
			xscale=.9,yscale=.9,
			vertex/.style={draw,circle, inner sep=.0pt, minimum size=0.53cm},font=\scriptsize]
			
			\node[vertex] (v) at (-1.41,0) {$v$};
			\node[vertex] (v') at (1.41,0) {$v'$};
			
			\draw (v) -- (v');
			\draw[white] (0,0) ellipse (1cm and 1.5cm);
			
			\draw[dashed] (0,0) ellipse (2.2cm and 1cm);
			\clip (0,0) ellipse (2.2cm and 1cm);
			
			\draw[] (v) -- +(130:2);
			\draw[] (v) -- +(200:2);
			
			\draw[] (v') -- +(-60:2);
			\draw[] (v') -- +(-10:2);
			\draw[] (v') -- +(100:2);
			\end{tikzpicture}
			\captionsetup{width=1\linewidth}
			\subcaption{An edge to be stabilized.}
		\end{subfigure}
		\begin{subfigure}[c]{0.59\textwidth}
			\begin{tikzpicture}[
			xscale=.9,yscale=.9,
			vertex/.style={draw,circle, inner sep=.0pt, minimum size=0.53cm},font=\scriptsize]
			
			\node[vertex] (v) at (-3,0) {$v$};
			\node[vertex] (u1) at (22.5+3*45:1.2) {$u_1$};
			\node[vertex] (u2) at (22.5+2*45:1.2) {$u_2$};
			\node[vertex] (u3) at (22.5+1*45:1.2) {$u_3$};
			\node[vertex] (u4) at (22.5+0*45:1.2) {$u_4$};
			\node[vertex] (u1') at (22.5+4*45:1.2) {$u'_1$};
			\node[vertex] (u2') at (22.5+5*45:1.2) {$u'_2$};
			\node[vertex] (u3') at (22.5+6*45:1.2) {$u'_3$};
			\node[vertex] (u4') at (22.5+7*45:1.2) {$u'_4$};
			\node[vertex] (v') at (3,0) {$v'$};
			
			\draw (v) -- (v');	
			
			\draw (u1) -- (v) -- (u2);
			\draw (u3) -- (v') -- (u4);
			
			\draw (u2') -- (u1) -- (u2);
			\draw (u3') -- (u1) -- (u3);
			\draw (u4') -- (u1) -- (u4);
			
			\draw (u1') -- (u2) -- (u1);
			\draw (u3') -- (u2) -- (u3);
			\draw (u4') -- (u2) -- (u4);
			
			\draw (u1') -- (u3) -- (u1);
			\draw (u2') -- (u3) -- (u2);
			\draw (u4') -- (u3) -- (u4);
			
			\draw (u1') -- (u4) -- (u1);
			\draw (u2') -- (u4) -- (u2);
			\draw (u3') -- (u4) -- (u3);
			
			\draw (u1') -- (u2') -- (u3') -- (u4') -- (u1');
			\draw (u1') -- (u3');
			\draw (u2') -- (u4');
			
			\draw[dashed] (0,0) ellipse (4cm and 1.5cm);
			\clip (0,0) ellipse (4cm and 1.5cm);
			
			\draw[] (v) -- +(130:2);
			\draw[] (v) -- +(200:2);
			
			\draw[] (v') -- +(-60:2);
			\draw[] (v') -- +(-10:2);
			\draw[] (v') -- +(100:2);
			
			\end{tikzpicture}
			\captionsetup{width=1\linewidth}
			\subcaption{The same section after adding the stabilization gadget.}
			\label{fig:TwoWayGadgetB}
		\end{subfigure}
		\caption{How to stabilize an arbitrary edge $\{v,v'\}$ without introducing new unfrozen edges.}\label{fig:TwoWayGadget}
	\end{center}
\end{figure} 

\setcounter{tempcounter}{\value{theorem}}
\setcounterref{theorem}{lem:quadripartite-gadget}
\addtocounter{theorem}{-1}
\begin{lemma}
	Let a nonempty graph $G$ and an edge $e\in V(G)$ be given. Construct from $G$ a graph $G'$ by substituting for $e$ the constant-size gadget that consists of a clique on the new vertex set $\{u_1,u_2,u_3,u_4,u_1',u_2',u_3',u_4'\}$, with the four edges $\{u_i,u_i'\}$ for $i\in\{1,2,3,4\}$ removed and the four edges $\{v,u_1\}$, $\{v,u_2\}$, $\{v',u_3\}$, and $\{v',u_4\}$ added. (This gadget is displayed in Figure~\ref{fig:TwoWayGadgetB}.) 
	The graph $G'$ has the following properties. 
\begin{enumerate}
	\item 
	$\beta(G')=\beta(G)+6$, 
	\item 
	every edge $e'\in E(G)\setminus\{e\}$ is $\beta$-stable in $G$ exactly if it is in $G'$, 
	\item 
	all remaining edges 
	of $G'$ 
	are $\beta$-stable,
	\item 
	every nonedge $e'\in\coE(G)$ is $\beta$-unfrozen in $G$ exactly if it is in $G'$, and
	\item 
	all remaining nonedges 
	$e'\in\coE(G')\setminus\coE(G)$ 
	of $G'$ 
	are $\beta$-unfrozen. 
\end{enumerate} 
\end{lemma}
\setcounter{theorem}{\value{tempcounter}}

\begin{proof}
	Let a graph $G$ and an edge $\{v,v'\}\in E(G)$ in it be given. 
	Let $Q
	$ be the quadripartite graph that is the join of the four empty graphs with vertex sets $\{u_1,u_1'\}$, $\{u_2,u_2'\}$, $\{u_3,u_3'\}$, and $\{u_4,u_4'\}$. 
	Let $G'$
	be the disjoint union of $G$ and $Q$ with the added edges $\{v,u_1\}$, $\{v,u_2\}$, $\{v',u_3\}$, and $\{v',u_4\}$. 
	See Figure~\ref{fig:TwoWayGadget} for a depiction of this construction. 
	
	We prove that $G'$ has the required properties. 
	\begin{enumerate}
		\item 
		Let $X$ be a vertex cover of $G$. It must contain $v$ or $v'$. If $v\in X$, then it follows that  $X\cup\{u_2,u_3,u_4,u_2',u_3',u_4'\}$ is a vertex cover of $G'$; if $v'\in X$, then $X\cup\{u_1,u_2,u_3,u_1',u_2',u_3'\}$ is one. This proves $\beta(G')\le \beta(G)+6$. 
		To obtain the inverse inequality, let $X'$ be a vertex cover of $G'$. Then $X'\setminus V(Q)$ is a vertex cover of $G$. Moreover, for any vertex $w\in V(Q)$, we have that, if $w\notin X'$, then $X'$ must contain the entire neighborhood of $w$, which contains exactly six vertices from $Q$. It follows that $\beta(G)\le \card{X'\setminus V(Q)}\le\card{X'}-6\le \beta(G')-6$. 
		\item 
		This is a consequence of the first property since, 
		for every edge $e'\in E(G)\setminus\{e\}$, our construction clearly commutes with deleting $e'$. 
		\item 
		Let $e'=e$ or $e'\in E(G')\setminus E(G)$. We need to show that $\beta(G'-e')\ge \beta(G)+6$. 
		It suffices to note that the complement graph of $G'-e'$ restricted to $V(Q)$ has no clique of size 3. 
		\item 
		The argument for the second property is valid for nonedges $e'\in\coE(G)$ as well.
		\item 
		Let $e'\in\coE(G')\setminus\coE(G)$ and let $X$ be a vertex cover of $G$. We show how to obtain a vertex cover for $G'+e'$ by adding six vertices to $X$. At least one endpoint of $e$ lies in $V(Q)$, call it $w$. If $w\in \{u_1,u_1',u_4,u_4'\}$, let $X'=X\cup\{u_1,u_1',u_4,u_4'\}$; otherwise, let $X'=X\cup\{u_2,u_2',u_3,u_3'\}$. 
		Let $X''=X'\cup\{u_3,u_3',u_4,u_4'\}$ if $v\in X$. Otherwise, we have $v'\in X$ and let $X''=X'\cup\{u_1,u_1',u_2,u_2'\}$. 
		It is easy to check that $X''$ is a vertex cover of $G'+e'$ and $\card{X''}=\card{X}+6$ in all cases.
\end{enumerate}
This concludes the proof. 
\end{proof}

\section{Proof of Theorem~\ref{thm:beta-twowaystability}}
\label{app:beta-twowaystability}
We restate and prove Theorem~\ref{thm:beta-twowaystability}.
\setcounter{tempcounter}{\value{theorem}}
\setcounterref{theorem}{thm:beta-twowaystability}
\addtocounter{theorem}{-1}
\begin{theorem}
	The problem $\beta$-\TwoWayStability is \ThetaTwo-complete.
\end{theorem}
\begin{proof}
	The upper bound is immediate. We now give a polynomial-time many-one reduction from $\beta$-\Unfrozenness, which is \ThetaTwo-hard by Theorem~\ref{thm:beta-unfrozenness}, to $\beta$-\TwoWayStability. 
	For given $G$, we replace each edge $e\in E(G)$ by the gadget displayed in Figure~\ref{fig:TwoWayGadgetB} and call the resulting graph $\widehat{G}$.  
	This is possible in polynomial time because the gadget has constant size. 
	By an iterated application of Lemma~\ref{lem:quadripartite-gadget}, all new edges in the resulting graph $\widehat{G}$ are $\beta$-stable and each pre-existing edge $e\in E(G)$ is $\beta$-unfrozen in $\widehat{G}$ if and only if it was $\beta$-unfrozen in $G$. Thus $\widehat{G}$ is $\beta$-two-way-stable if and only if $G$ is $\beta$-unfrozen.
\end{proof}
\setcounter{theorem}{\value{tempcounter}}
\end{appendix}

\end{document}